\documentclass[11pt,letterpaper]{article}

\usepackage[margin=1in]{geometry}

\usepackage{amsfonts, amsmath, amssymb, amsthm}
\usepackage{mathrsfs} % \mathscr{P} or \mathscr{B}
\usepackage{setspace}

\usepackage{hyperref}
\hypersetup{
    colorlinks=true,
    linkcolor=teal,
    citecolor=teal,
    urlcolor=purple,
}
\usepackage{graphicx}           
\usepackage{xcolor}
\usepackage{algorithm}
\usepackage{algorithmic}
\usepackage{subfig}

\usepackage[]{natbib} % to use \citep, \citet, etc.
\usepackage{bibunits}
\defaultbibliographystyle{plainnat}
\defaultbibliography{ref}

\newtheorem{theorem}{Theorem} % \newtheorem{theorem}{Theorem}[section] for numbering by sections		
\newtheorem{lemma}{Lemma} % \newtheorem{lemma}[theorem]{Lemma} to follow Theorem numbering

\theoremstyle{definition}
\newtheorem{definition}{Definition}	
\newtheorem{remark}{Remark}

\newcommand{\cC}{\mathcal{C}}

\newcommand{\cI}{\mathcal{I}}
\newcommand{\cM}{\mathcal{M}}

\newcommand{\cS}{\mathcal{S}}

\newcommand{\cX}{\mathcal{X}}
\newcommand{\cY}{\mathcal{Y}}

\newcommand{\N}{\mathbb{N}}
\newcommand{\R}{\mathbb{R}}

\newcommand{\E}{\mathbb{E}} % no spacing
\renewcommand{\P}{\mathbb{P}}

\DeclareMathOperator*{\argmin}{arg\,min} % \, for spacing
\DeclareMathOperator*{\argmax}{arg\,max} % \, for spacing

\renewcommand{\epsilon}{\varepsilon}
\DeclareMathOperator*{\motimes}{\text{\raisebox{0.25ex}{\scalebox{0.7}{$\bigotimes$}}}}

\title{Inference for Clustering: Conformal Sets for Cluster Labels}
\author{YoonHaeng Hur$^\ast$ \\
    Department of Statistics, Columbia University\\
    and \\
    Anirban Nath$^\ast$\\
    Department of Statistics, Columbia University\\
    and \\
    Genevera Allen \\
    Department of Statistics, Columbia University}

\begin{document}

\maketitle

\begin{abstract}
    While clustering is ubiquitously used across science and industry, uncertainty in cluster assignments is rarely quantified with rigorous guarantees. We propose a novel conformal inference framework for clustering that returns confidence sets for cluster labels. The key challenge is that labels are unobserved and estimated from data, so naively using deterministic cluster labels can violate exchangeability and induce severe under-coverage. To address this, we propose split conformal clustering with stochastic labels, which samples labels from soft cluster labels, fits a soft classifier to predict these stochastic labels, and calibrates conformal scores to construct confidence sets for cluster labels at any query point. We establish a finite-sample lower bound on marginal coverage that reveals how under-coverage is controlled by two properties of the clustering algorithm: consistency of estimated soft labels and replace-one stability. Under mild conditions, we prove asymptotic coverage and verify these conditions for correctly specified parametric mixture models. Simulations for mixture models show that our method attains target coverage with informative set sizes, validating our theoretical results. Applications to clustering cell types in single-cell RNA-seq data demonstrate the practical utility and interpretability of our approach to quantifying cluster label uncertainty. 
\end{abstract}

\begin{bibunit}
\section{Introduction}
\label{sec:intro}
Clustering is a fundamental unsupervised learning task that aims to discover latent group structure in data and is widely used across scientific and industrial domains, including biomedicine \citep{lopez2018unsupervised}, community detection in social network analysis \citep{zhou2019analysis}, marketing \citep{kansal2018customer}, among many others. Over several decades, a rich literature has developed around clustering methodology, theory, and computation, resulting in a wide range of algorithms and software tools favored in practice. Despite this maturity, clustering analyses are often brittle, meaning that small perturbations of the data or changes in preprocessing can lead to irreproducible scientific findings and unreliable downstream decisions \citep{senbabaouglu2014critical, smith2025lack}, particularly in single-cell analyses where clustering is routinely used for putative cell-type discovery and is known to be sensitive to preprocessing and method choice \citep{wang2020impact, duo2020systematic}. 

A natural response to this instability is to ask for uncertainty quantification in clustering. However, what “uncertainty” means in the context of clustering is inherently ambiguous. Broadly speaking, uncertainty in clustering can be viewed from several perspectives, including uncertainty in whether clusters exist at all, uncertainty in the locations or shapes of cluster centroids, and uncertainty in the labels produced by a clustering algorithm. Uncertainty in centroids seems intuitive but is only applicable to centroid-based clustering methods, and while practical approaches have been suggested before \citep{kerr2001bootcluster, hofmans2015bootstrapkmeans,liu2018ukmeans}, we are not aware of theoretically-grounded uncertainty quantification of this type. Uncertainty in the existence of clusters has been studied and was first proposed by \citet{liu2008statistical} as SigClust, which tests for the existence of clusters against a null of a single Gaussian cluster; this approach has been extended to other settings \citep{huang2015statistical,kimes2017statistical,shen2024statistical}. Addressing similar types of uncertainty quantification, several works have recently studied clustering from a selective inference perspective to provide inference on the existence of clusters \citep{yun2023selective,chen2023selective,gao2024selective}. From a Bayesian perspective, uncertainty is often quantified over the entire partition through credible sets or credible balls under mixture models, providing an explicit ``region of plausible clusterings” rather than just labels \citep{wade2018bayesian,dahl2022search}. Thus, the existing literature is largely concentrated on cluster existence or model-specific structural uncertainty, and does not directly address uncertainty in the cluster labels themselves. In contrast, our goal is to quantify the uncertainty in the cluster labels, ideally in a model-agnostic and distribution-free manner.

There is currently no broadly applicable framework that provides valid, finite-sample inference for the output labels of a clustering algorithm in a model-agnostic manner. Uncertainty in cluster labels is fundamentally different from uncertainty in centroids or population-level structure, as cluster labels are discrete and permutation-invariant. Accordingly, label assignments cannot be treated as smooth or parametric objects whose variability is summarized by classical estimation error. Moreover, most clustering algorithms are deterministic and are not designed to produce meaningful out-of-sample predictions. Hence, a partition of the observed data returned by standard procedures does not naturally define a predictive rule for assigning labels to new points in a way that reflects uncertainty. Common ad hoc extensions—such as nearest-centroid assignments or auxiliary classifiers trained on cluster labels—are algorithm-dependent and lack formal inferential guarantees. As a result, cluster labels are tightly coupled to the specific dataset on which they are computed, limiting their interpretability and generalizability beyond the observed sample. These considerations suggest that uncertainty in clustering is better understood not as a property of a fixed sample partition, but as a property of a labeling rule defined over the feature space. This perspective motivates the development of a meta-algorithm that takes an arbitrary clustering procedure as input and produces confidence sets for cluster membership at any point in the data domain. Such prediction sets provide a natural and interpretable representation of uncertainty for any data point in the feature space, identifying regions where cluster assignment is reliable and regions where multiple labels remain plausible.

Recent developments in conformal inference offer a promising foundation for such an approach. Conformal inference provides distribution-free, finite-sample guarantees for predictive uncertainty with only the exchangeability assumption \citep{vovk2005algorithmic}. Most of the conformal inference literature has focused on regression \citep{lei2018distribution,romano2019conformalized} and classification \citep{lei2014classification,sadinle2019least,romano2020classification}, while conformal inference in unsupervised learning has received relatively little attention. The most prominent example in this scarce literature would be \citet{lei2013distribution}, which constructs distribution-free prediction sets for density estimation; these sets can be cut at a certain level to yield density-based or modal clusters, which are later extended to functional data and other domains \citep{lei2015conformal,adams2025functionalcad}. However, these approaches are not broadly applicable and are limited to settings where kernel density estimation performs well. Also, there has been a line of work characterizing outliers via conformal p-values \citep{bates2023testing,liang2024integrative,lee2025full}, with applications to streaming trajectories \citep{laxhammar2015inductive} and time series \citep{ishimtsev2017conformal,safin2017conformal}. While there have been related approaches under the name of ``conformal clustering'' \citep{cherubin2015conformal,nouretdinov2020multi,kiani2020conformalized}, we should note that they use conformal p-values to determine outliers and then cluster these outliers into groups, which has nothing to do with providing uncertainty quantification for clustering results. 

In this paper, we develop a principled framework for uncertainty quantification of clustering labels based on conformal inference, overcoming the unique challenges posed by the unsupervised nature of clustering. We construct confidence sets for cluster labels that can achieve valid marginal coverage, even when labels are estimated through an arbitrary clustering algorithm. Since the cluster labels are not observed but are instead estimated from the data itself, conformal procedures based on the estimated labels induce a form of distribution shift and break the exchangeability condition required for standard conformal guarantees. By employing stochastic clustering methods, we tackle the dependence induced by the estimated labels and provide a rigorous characterization of coverage guarantees. This leads to an interpretable and actionable procedure for practitioners, allowing them to identify regions of the data domain where cluster assignments are reliable and regions where multiple labels remain plausible. 

\begin{figure}[!htbp]
    \centering
    \includegraphics[width=0.7\textwidth]{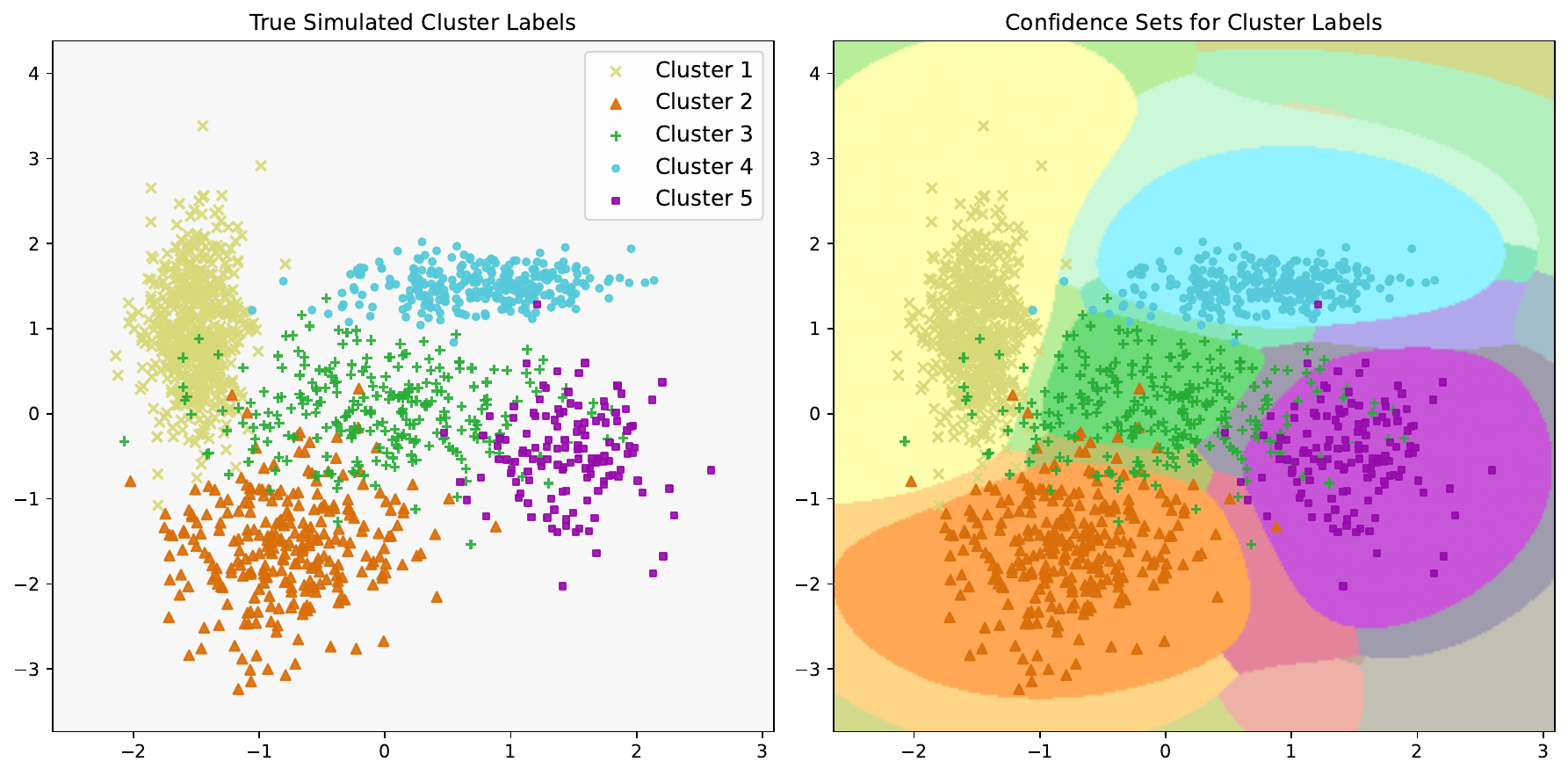}
    \caption{The left scatter plot shows $n = 1500$ data points with true cluster labels from a simulated example of a mixture of $K = 5$ Gaussian distributions on $\R^2$. The right panel shows the proposed $95\%$ confidence sets for cluster labels (Algorithm \ref{alg:split_conformal_clustering_stochastic} with stochastic GMM clustering), visualized as a heatmap over the data domain; colors correspond to those of the true cluster labels (same as the left plot), with in-between colors denoting regions where confidence sets contain two or more cluster labels.}
    \label{fig:motivating_figure}
\end{figure}

To motivate our approach, consider an illustrative simulation in Figure \ref{fig:motivating_figure}. Here, our valid confidence sets (Algorithm \ref{alg:split_conformal_clustering_stochastic} in Section \ref{sec:conformal_clustering_stochastic}) successfully highlight regions in the data domain where we are certain of the cluster labels, while the regions of uncertainty---the boundaries between clusters---that might correspond to two or more cluster labels at $95\%$ confidence are shown in-between colors. Overall, this work develops a unique framework leveraging conformal inference in unsupervised learning, one of the first of its kind in the literature, that addresses an important target of inference and develops a method that
has the potential for major improvement in the reliability and interpretability of clustering results.

The rest of the paper is organized as follows. Section \ref{sec:method} describes the proposed method for constructing conformal sets for cluster labels; we first describe two naive approaches and explain why they fail to provide valid coverage guarantees, which motivates our proposed method based on stochastic clustering. Section \ref{sec:theory} provides theoretical results on the coverage guarantees of our method, and Section \ref{sec:simulations} presents simulation studies and applications to single-cell data, which demonstrate the practical performance of our method. Section \ref{sec:discussion} concludes with a discussion of limitations and future directions. All technical proofs are deferred to the Supplementary Material, together with additional empirical results and further discussion on related literature.
\section{Conformal Sets for Cluster Labels}
\label{sec:method}
\paragraph*{Notation} For $n \in \N$, let $[n] = \{1, \ldots, n\}$, denote by $\cS_n$ the set of all permutations of $[n]$, and define $\Delta_n = \{w \in \R_+^n : \sum_{i = 1}^{n} w_i = 1\}$; also, for any $\sigma \in \cS_n$ and $C \subset [n]$, let $\sigma(C) = \{\sigma(c) : c \in C\}$. For $w \in \Delta_n$, let $\mathrm{Cat}(w)$ denote the categorical distribution on $[n]$ that draws $i \in [n]$ with probability $w_i$. For $u, v \in \R^K$, let $\|u - v\|_1 = \sum_{k = 1}^{K} |u_k - v_k|$. For $u, v \in \Delta_K$, let $H^2(u, v) = \frac{1}{2} \sum_{k = 1}^{K} (\sqrt{u_k} - \sqrt{v_k})^2$ denote the square of the Hellinger distance. Also, $\xrightarrow{p}$ denotes the convergence in probability, $o_p(1)$ denotes a sequence of random variables converging to zero, and $O_p(1)$ denotes a sequence bounded in probability.

\subsection{Problem Formulation}
\label{sec:problem_formulation}
Given $X_1, \ldots, X_n \in \R^p$, clustering groups the data into clusters with labels $Y_1, \ldots, Y_n \in [K]$, where $K$ is the number of clusters. Here, we assume $K$ is fixed and known. We postulate a probabilistic setting where these covariates are in fact generated with true labels, say, $(X_1, Y_1^\ast), \ldots, (X_n, Y_n^\ast)$ independently drawn from some distribution $P^\ast$ on $\R^p \times [K]$, where the true cluster labels $Y_i^\ast$'s are unobserved. 

Our goal is to construct some confidence set, $\hat{\cC}(x) \subset \{1, \ldots, K\}$ for any $x \in \R^p$, such that the set $\hat{\cC}(X_{n + 1})$ contains the true label $Y_{n + 1}^\ast$ with high probability for a new test point $(X_{n + 1}, Y_{n + 1}^\ast) \sim P^\ast$ independent of $\{(X_i, Y_i^\ast)\}_{i = 1}^{n}$. To rigorously state this, however, we note that cluster labels are invariant to permutations. Hence, we define the oracle cluster label permutation as $\hat{\sigma}_o^\ast = \argmax_{\sigma \in \cS_K} \P\left(\sigma(Y^\ast) \in \hat{\cC}(X) \, \big| \, \hat{\cC}\right)$, where $(X, Y^\ast) \sim P$ is independent of the randomness of $\hat{\cC}$.\footnote{The main source of the randomness in $\hat{\cC}$ is $\{X_i\}_{i = 1}^{n}$, but it may include additional randomness depending on the clustering algorithm, as we will see in Section \ref{sec:conformal_clustering_stochastic}.} Thus, $\hat{\sigma}_o^\ast$ is the oracle alignment of $\hat{\cC}$ to the ground truth labels evaluated on the new data $(X, Y^\ast)$ independent of the construction of $\hat{\cC}$. Now, with this oracle permutation, our goal is to construct a $\hat{\cC}$ such that
\begin{equation}
    \label{eq:target_coverage_bound}
    \P\left(\hat{\sigma}_o^\ast(Y_{n + 1}^\ast) \in \hat{\cC}(X_{n + 1})\right) \ge 1 - \alpha
\end{equation}
for a given level $\alpha \in (0, 1)$. In other words, our goal is to construct a confidence set that, with at least $1 - \alpha$ probability, will cover the true, unknown cluster labels. These confidence sets will help us quantify the uncertainty in cluster labels. For instance, if $\hat{\cC}(x)$ contains multiple labels, it indicates that the cluster label for $x$ is uncertain and could be any of those labels in $\hat{\cC}(x)$; if $\hat{\cC}(x)$ contains only one label, it signals that the cluster label for $x$ is more certain.

\subsection{Naive Approach I: Cutoff for Generalizable Clustering}
\label{sec:cutoff}
For certain types of clustering algorithms, such as parametric mixture model clustering, we can obtain the soft cluster label for any input $x \in \R^p$, namely, the posterior probability of $x$ belonging to each cluster. Then, it is natural to ask whether we can directly use these soft labels to construct confidence sets for cluster labels. Particularly, one can simply define the confidence set with cluster labels having the largest probability so that the cumulative probability exceeds $1 - \alpha$. We formally call such clustering algorithms generalizable soft clustering and summarize this method---called cutoff---in Algorithm \ref{alg:cutoff}.

\begin{definition}
    \label{def:generalizable_soft_clustering}
    We call $\gamma$ a generalizable soft clustering algorithm if it takes a sequence $x_1, \ldots, x_n \in \R^p$ as input and outputs a function $\hat{\gamma}_n \colon \R^p \to \Delta_K$ such that $\hat{\gamma}_n(x) = (\hat{\gamma}_n(x)_1, \ldots, \hat{\gamma}_n(x)_K)$ is a probability vector representing the soft cluster label for $x$, where the $k$-th entry $\hat{\gamma}_n(x)_k$ denotes the probability of $x$ belonging to cluster $k$.
\end{definition}

\renewcommand{\thealgorithm}{0.\arabic{algorithm}} % Force 0.x format
\begin{algorithm}[!htbp]
    \caption{Cutoff Set for Generalizable Soft Clustering}
    \label{alg:cutoff}
    \begin{algorithmic}[1]
        \REQUIRE Unlabeled data $X_1, \ldots, X_n \in \R^p$ and user-specified level $\alpha \in (0, 1)$.
        \REQUIRE Generalizable soft clustering $\gamma$.
        \STATE Implement $\gamma$ with input $X_1, \ldots, X_n$: obtain $\hat{\gamma}_n \colon \R^p \to \Delta_K$.
        \STATE For the soft label $\hat{\gamma}_n(x)$, let $\hat{\gamma}_n(x)_{(1)} \ge \cdots \ge \hat{\gamma}_n(x)_{(K)}$ be the order statistics of the entries of $\hat{\gamma}_n(x)$, with the corresponding permutation $\mathrm{rk}_x \in \cS_K$ denoting the ranks so that $\hat{\gamma}_n(x)_k = \hat{\gamma}_n(x)_{(\mathrm{rk}_x(k))}$ for $k \in [K]$.
        \STATE Define the cutoff threshold $\mathrm{cut}_x := \min\left\{k \in [K] : \sum_{\ell = 1}^{k} \hat{\gamma}_n(x)_{(\ell)} \ge 1 - \alpha\right\}$ and \vspace*{-10pt}
        \begin{equation*}
            \hat{\cC}(x) = \Big\{y \in [K] : \mathrm{rk}_x(y) \le \mathrm{cut}_x\Big\}. \vspace*{-25pt}
        \end{equation*}
        \RETURN $\hat{\cC}$.
    \end{algorithmic}
\end{algorithm}

Despite its simplicity, it turns out that Algorithm \ref{alg:cutoff} is not guaranteed to achieve the desired coverage guarantee \eqref{eq:target_coverage_bound}; a similar phenomenon for soft classifiers is discussed in \citet{angelopoulos2020uncertainty}. Intuitively, the soft labels obtained from the clustering algorithm are estimated from the data and may not be well calibrated, which can lead to under-coverage when the sample size is small. On the other hand, when the sample size is large enough, it is prone to over-coverage with uninformative confidence sets, larger than we would like. To see this, imagine a point whose largest soft label is $0.8$, which is indeed the true cluster label, but there are four other soft labels that are all $0.025$. Despite the strong signal in the largest soft label, the cutoff method has to include all four labels in the confidence set to achieve the target coverage of $0.9$, which inflates the set size more than necessary. We will confirm this intuition both theoretically and empirically in later sections.

\subsection{Background: Conformal Prediction}
\label{sec:background_conformal_prediction}
What is missing for the above cutoff method is a rigorous way to calibrate the soft labels obtained from the clustering algorithm, which is crucial to achieve the desired coverage guarantee with informative confidence sets. Also, it is only applicable to generalizable soft clustering algorithms, which limits its applicability.
To address these issues, we take inspiration from conformal prediction for classification, which provides a general framework for constructing distribution-free confidence sets with guaranteed coverage. Notice that once we have the true labels $Y_i^\ast$'s, the clustering problem stated in Section \ref{sec:problem_formulation} can be viewed as a classification problem with covariates $X_i$'s and labels $Y_i^\ast$'s. Conformal classification constructs prediction set $\hat{\cC}(x) \subset [K]$ for any $x \in \R^p$ such that given a level $\alpha \in (0, 1)$, we have $\P(Y_{n + 1}^\ast \in \hat{\cC}(X_{n + 1})) \ge 1 - \alpha$, where $(X_{n + 1}, Y_{n + 1}^\ast) \sim P^\ast$ is a new independent test point.

Split conformal prediction is the standard routine, which randomly divides labeled data into training and calibration samples, fits a soft classifier $\hat{\pi} \colon \R^p \to \Delta_K$ on the training set, and uses the calibration sample to calibrate conformity scores. One widely used score is the generalized inverse quantile score of \citet{romano2020classification}: for $(x, y) \in \R^p \times [K]$, define $s((x, y); \hat{\pi}) = \hat{\pi}(x)_{(0)} + \cdots + \hat{\pi}(x)_{(r(y) - 1)}$, where $\hat{\pi}(x)_{(1)} \ge \cdots \ge \hat{\pi}(x)_{(K)}$ are the order statistics of the entries of $\hat{\pi}(x)$, and $r(y)$ is the rank of the $y$-th entry among the entries of $\hat{\pi}(x)$, and $\hat{\pi}(x)_{(0)} := 0$ for convenience.\footnote{The original definition of the generalized inverse quantile score in \citet{romano2020classification} incorporates an additional randomization term, but we omit it here for simplicity.} The resulting prediction set contains labels whose scores fall below an empirical quantile from the calibration sample. Validity of conformal inference follows from the exchangeability of the labeled data, regardless of the choice of the classifier or conformity scores. As we will see, however, exchangeability is violated in the clustering setting because the cluster labels have to be estimated from the data. 

In this regard, one may naturally draw connections to the recent developments in the conformal classification literature with noisy labels \citep{einbinder2024labelnoisejmlr,sesia2025adaptivejrssb,feldman2025conformal}. However, these works rely critically on the assumption that the label corruption mechanism is (conditionally) independent of the covariates. This assumption fails in our setting, where the labels are estimated from the data and hence are inherently covariate-dependent. Likewise, developments in conformal prediction for semi-supervised learning \citep{zhou2025semi} are not directly applicable to our framework.

\subsection{Naive Approach II: Split Conformal Clustering}
\label{sec:conformal_clustering_naive}
We now describe a naive approach for constructing confidence sets for cluster labels, which we call split conformal clustering, summarized in Algorithm \ref{alg:split_conformal_clustering}. It follows the general idea of conformal classification, but with the twist that we do not have access to the true labels. Instead, we obtain cluster labels from the data by adopting the idea of generalizability that leverages methods of clustering prediction strength \citep{tibshirani2005cluster,lange2004stability}. Concretely, we start by randomly splitting the data into a training and calibration set and then cluster each of these sets separately to obtain cluster labels. Then, on the training set, we propose to treat the cluster labels as if they were the truth and build a classifier on the training set to predict soft labels on the calibration set. The classifier allows for prediction on any input point, including the calibration points and the new test point. Since the cluster labels from the training and calibration sets are obtained separately, they may differ by a permutation. Hence, we align the predictions and cluster labels on the calibration set by finding the best permutation that minimizes the disagreement between the labels; this is a linear assignment problem that can be solved in polynomial time \citep{burkard_dellamico_martello_2009}. Finally, we use conformity scores to compare the calibration cluster labels to the calibration soft predictions and construct the confidence set, including those labels whose scores are below the appropriate quantile of the calibration scores.

\begin{algorithm}[!htbp]
    \caption{Split Conformal Clustering (Naive)}
    \label{alg:split_conformal_clustering}
    \begin{algorithmic}[1]
    \REQUIRE Unlabeled data $X_1, \ldots, X_n \in \R^p$ and user-specified level $\alpha \in (0, 1)$.
    \STATE Split the data into $\{X_i\}_{i \in \cI_{tr}}$ and $\{X_i\}_{i \in \cI_{ca}}$ for some partition $\cI_{tr} \cup \cI_{ca} = [n]$.
    \STATE Cluster the training data $\{X_i\}_{i \in \cI_{tr}}$ and obtain the corresponding labels $\{Y_i\}_{i \in \cI_{tr}}$.
    \STATE Fit a soft classifier $\hat{\pi} \colon \R^p \to \Delta_K$ on the cluster-labeled training data $\{(X_i, Y_i)\}_{i \in \cI_{tr}}$.
    \STATE Cluster the calibration data $\{X_i\}_{i \in \cI_{ca}}$ and obtain the corresponding labels $\{Y_i\}_{i \in \cI_{ca}}$.
    \STATE Align the cluster and classification labels: $\hat{\sigma} \in \argmin_{\sigma \in \cS_K} ~ \sum_{i \in \cI_{ca}} 1\{\hat{f}(X_i) \neq \sigma(Y_i)\}$, where $\hat{f}(x) = \argmax_{k \in [K]} \hat{\pi}(x)_k$ is the classification rule based on $\hat{\pi}$. 
    \STATE Define suitable conformity scores $s_i := s((X_i, \hat{\sigma}(Y_i)); \hat{\pi}) \in \R$ for $i \in \cI_{ca}$ and construct \vspace*{-10pt}
    \begin{equation*}
        \hat{\cC}(x) = \Big\{y \in [K] : s((x, y); \hat{\pi}) \le \lceil (1 - \alpha) (1 + |\cI_{ca}|)\rceil \text{-th smallest value of} ~ \{s_i\}_{i \in \cI_{ca}}\Big\}. \vspace*{-15pt}
    \end{equation*}
    \RETURN $\hat{\cC}$.
\end{algorithmic}
\end{algorithm}

Steps 2 and 4 of Algorithm \ref{alg:split_conformal_clustering} require a clustering algorithm to assign labels on the training and calibration sets. Although any clustering algorithm can in principle be used, standard methods that produce hard labels are inadequate. Their deterministic assignments do not reflect the true uncertainty in cluster membership and can therefore lead to under-coverage, which we show theoretically in the subsequent section. Thus, we view Algorithm \ref{alg:split_conformal_clustering} as a naive baseline, and introduce our main algorithm in the next section.

\subsection{Our Approach: Split Conformal Clustering with Stochastic Labels}
\label{sec:conformal_clustering_stochastic}
The limitations of the above two naive approaches clearly manifest the challenges in constructing confidence sets for cluster labels, motivating us to develop well-calibrated methods that can better capture the uncertainty in cluster labels. To this end, we consider stochastic clustering, which seeks to insert the true uncertainty in the cluster labels back into the labeling mechanism. By stochastic clustering, we refer to a clustering algorithm that incorporates randomness in the clustering process such that the cluster labels are not deterministic but rather sampled from some distribution. More precisely, from any soft clustering algorithm that assigns probability vectors to the input points, such as parametric mixture model clustering or fuzzy-c-means \citep{dunn1973fuzzy,bezdek1984fcm}, we can obtain a stochastic clustering algorithm by sampling the cluster labels from these assigned probability vectors. 

\begin{definition}
    \label{def:stochastic_clustering}
    We call $\gamma$ a stochastic clustering algorithm if it takes a sequence $x_1, \ldots, x_n \in \R^p$ as input, obtains probability vectors $\hat{\gamma}_n(x_1), \ldots, \hat{\gamma}_n(x_n) \in \Delta_K$, and outputs the corresponding labels by sampling $y_i \sim \mathrm{Cat}(\hat{\gamma}_n(x_i))$ for $i \in [n]$ independently.\footnote{Technically, $\hat{\gamma}_n$ is a map from $\{x_1, \ldots, x_n\}$ to $\Delta_K$. If we use generalizable soft clustering (Definition \ref{def:generalizable_soft_clustering}), then $\hat{\gamma}_n$ is a map from $\R^p$ to $\Delta_K$, generalizing outside of $\{x_1, \ldots, x_n\}$.}
\end{definition}

\setcounter{algorithm}{0} % Reset algorithm counter to 0
\renewcommand{\thealgorithm}{\arabic{algorithm}} % Restore normal numbering format
\begin{algorithm}[!htbp]
    \caption{Split Conformal Clustering with Stochastic Labels}
    \label{alg:split_conformal_clustering_stochastic}
    \begin{algorithmic}[1]
    \REQUIRE Unlabeled data $X_1, \ldots, X_n \in \R^p$ and user-specified level $\alpha \in (0, 1)$.
    \STATE Split the data into $\{X_i\}_{i \in \cI_{tr}}$ and $\{X_i\}_{i \in \cI_{ca}}$ for some partition $\cI_{tr} \cup \cI_{ca} = [n]$.
    \STATE Apply stochastic clustering to the training data $\{X_i\}_{i \in \cI_{tr}}$ and sample the corresponding labels $\{Y_i\}_{i \in \cI_{tr}}$. 
    \STATE Fit a soft classifier $\hat{\pi} \colon \R^p \to \Delta_K$ on the cluster-labeled training data $\{(X_i, Y_i)\}_{i \in \cI_{tr}}$.
    \STATE Apply stochastic clustering to the calibration data $\{X_i\}_{i \in \cI_{ca}}$ and sample the corresponding labels $\{Y_i\}_{i \in \cI_{ca}}$. 
    \STATE Align the cluster and classification labels: $\hat{\sigma} \in \argmin_{\sigma \in \cS_K} ~ \sum_{i \in \cI_{ca}} 1\{\hat{f}(X_i) \neq \sigma(Y_i)\}$, where $\hat{f}(x) = \argmax_{k \in [K]} \hat{\pi}(x)_k$ is the classification rule based on $\hat{\pi}$.
    \STATE Define suitable conformity scores $s_i := s((X_i, \hat{\sigma}(Y_i)); \hat{\pi}) \in \R$ for $i \in \cI_{ca}$ and construct \vspace*{-10pt}
    \begin{equation*}
        \hat{\cC}(x) = \Big\{y \in [K] : s((x, y); \hat{\pi}) \le \lceil (1 - \alpha) (1 + |\cI_{ca}|)\rceil \text{-th smallest value of} ~ \{s_i\}_{i \in \cI_{ca}}\Big\}. \vspace*{-15pt}
    \end{equation*}
    \RETURN $\hat{\cC}$.
\end{algorithmic}
\end{algorithm}

Algorithm \ref{alg:split_conformal_clustering_stochastic} is a special case of Algorithm \ref{alg:split_conformal_clustering} where Steps 2 and 4 are implemented with a stochastic clustering algorithm. However, the stochasticity in the clustering process is crucial to capture the uncertainty in cluster labels. Since the sampled cluster labels better conform to the true cluster label distribution, we expect this approach to enjoy improved coverage compared to the deterministic cluster labels in Algorithm \ref{alg:split_conformal_clustering}. We confirm this key insight, both theoretically and empirically, in the later sections. 

\begin{remark}
    \label{rmk:classifier_necessity}
    One may wonder whether the classifier in Step 3 of Algorithm \ref{alg:split_conformal_clustering_stochastic} is necessary. In Algorithm \ref{alg:split_conformal_clustering} with hard clustering, this step is necessary to allow for prediction on any input point, including the calibration points and the new test point. However, if Algorithm \ref{alg:split_conformal_clustering_stochastic} is used with a generalizable soft clustering algorithm as defined in Definition \ref{def:generalizable_soft_clustering}, then we may skip the classifier by simply setting $\hat{\pi}(x) = \hat{\gamma}_{tr}(x)$ for any $x \in \R^p$, where $\hat{\gamma}_{tr}$ is fitted on $\{X_i\}_{i \in \cI_{tr}}$. That said, in practice, many commonly used stochastic clustering algorithms, such as fuzzy-c-means, may not satisfy this generalizability condition, so we need to train a separate classifier to allow for prediction on any input point, including the calibration points and the new test point. Even for generalizable clustering such as mixture model clustering, training a separate classifier can help improve the efficiency of the confidence sets by providing better soft predictions with more flexible cluster boundaries.
\end{remark}

\paragraph*{Practicalities} 
In practice, one needs to know the number of clusters $K$ before applying our split conformal stochastic clustering algorithm. However, choosing $K$ is a well-known challenge \citep{sugar2003finding}. We suggest that a practitioner should use well-established methods appropriate to their data to select $K$. These could include the Silhouette score \citep{Rousseeuw1987}, stability-based approaches \citep{ben2001stability}, or, more closely related to our framework, generalizability approaches \citep{tibshirani2005cluster,lange2004stability}. \citet{chang2025unsupervised} recently outlined a practical workflow for choosing $K$, among other items, and we refer the reader to this work to address how to choose $K$ or any other clustering hyperparameter in practice.

\section{Theoretical Guarantees}
\label{sec:theory}

For the confidence set $\hat{\cC}$ constructed by Algorithm \ref{alg:split_conformal_clustering_stochastic}, we first derive a finite-sample lower bound on the coverage $\P(\hat{\sigma}_o^\ast(Y_{n + 1}^\ast) \in \hat{\cC}(X_{n + 1}))$, which depends not only on the level $1 - \alpha$ but also on two properties of the input clustering algorithm: consistency and stability. The former characterizes the estimation error of the response probability of cluster labels given a covariate, while the latter represents how much the clustering result changes if we perturb one data point. After defining these concepts and establishing the coverage bound, we show that the desired coverage bound \eqref{eq:target_coverage_bound} holds asymptotically under mild conditions that are satisfied in well-specified settings. 

\paragraph*{Notation} Recall from Section \ref{sec:problem_formulation} that $P^\ast$ is the underlying distribution on $\R^p \times [K]$. Let $P_X^\ast$ be the marginal distribution on $\R^p$ and $\gamma^\ast(x) = (\P(Y^\ast = 1 \, | \, X = x), \ldots, \P(Y^\ast = K \, | \, X = x)) \in \Delta_K$ be the conditional probability vector of $Y^\ast$ given $X$ for $(X, Y^\ast) \sim P^\ast$.

\subsection{Coverage Bound}
We will show that the lower bound on the coverage $\P(\hat{\sigma}_o^\ast(Y_{n + 1}^\ast) \in \hat{\cC}(X_{n + 1}))$ depends on two properties of the clustering algorithm used in Algorithm \ref{alg:split_conformal_clustering_stochastic}: consistency and stability. To understand these concepts, notice that if the cluster-labeled data obtained in Algorithm \ref{alg:split_conformal_clustering_stochastic} approximately follow the underlying distribution $P^\ast$, Algorithm \ref{alg:split_conformal_clustering_stochastic} will perform like split conformal classification, which enjoys the usual $1 - \alpha$ guarantee. The deviation of the cluster-labeled data from $P^\ast$ can be measured by the estimation error of the response probability of cluster labels given a covariate and the stability of the clustering result when one data point is perturbed. We formalize these concepts in the following definitions.

\begin{definition}[Consistency of Clustering]
    \label{def:consistency}
    Define the estimation error of a stochastic clustering algorithm $\gamma$ given a sample of size $n$ to be 
    \begin{equation*}
        \mathsf{E}_n(\gamma) := \frac{1}{n} \sum_{i = 1}^{n} \E\left[\|\hat{\gamma}_n(X_i) - \gamma^\ast(X_i)\|_1\right],
    \end{equation*}
    where $\hat{\gamma}_n$ is fitted on $X_1, \ldots, X_n$. We say $\gamma$ is consistent if $\lim_{n \to \infty} \mathsf{E}_n(\gamma) = 0$.
\end{definition}

In Definition \ref{def:consistency}, the term $\|\hat{\gamma}_n(X_i) - \gamma^\ast(X_i)\|_1$ is the total variation distance between the obtained cluster label distribution $\gamma$ and the true label distribution $\gamma^\ast$, and $\mathsf{E}_n(\gamma)$ averages this over all inputs and takes the expectation. When the underlying distribution indeed has a well-separated cluster structure, a large sample will approximate the true distribution well, thereby having a small estimation error and thus small $\mathsf{E}_n(\gamma)$ for a certain choice of $\gamma$. In the next section, we will analyze cases where such consistency holds.

We also introduce replace-one stability. Conformal inference relies on exchangeability. If the cluster-labeled data together with the test point $(X_{n + 1}, Y_{n + 1}^\ast) \sim P^\ast$ are exchangeable, then Algorithm \ref{alg:split_conformal_clustering_stochastic} enjoys the usual $1 - \alpha$ coverage like split conformal classification. While exchangeability is not guaranteed in clustering settings, we can quantify the extent to which the cluster-labeled data are exchangeable by the above estimation error and the following stability term.

\begin{definition}[Replace-One Stability of Clustering]
    \label{def:replace-one_stability}
    Define the replace-one stability of a stochastic clustering algorithm $\gamma$ given a sample of size $n$ to be 
    \begin{equation*}
        \mathsf{S}_n(\gamma) := \frac{1}{n} \sum_{i = 1}^{n} \E\left[\Big\|\motimes_{j \in [n] \backslash \{i\}} \hat{\gamma}_n(X_j) - \motimes_{j \in [n] \backslash \{i\}} \hat{\gamma}_{i \to n + 1}(X_j)\Big\|_1\right],
    \end{equation*}
    where $\hat{\gamma}_n$ is fitted on $X_1, \ldots, X_n$ and $\hat{\gamma}_{i \to n + 1}$ is fitted on $X_1, \ldots, X_{i - 1}, X_{n + 1}, X_{i + 1}, \ldots, X_n$. We say $\gamma$ is asymptotically replace-one invariant if $\lim_{n \to \infty} \mathsf{S}_n(\gamma) = 0$.
\end{definition}

The replace-one stability essentially measures the change in the joint distributions of the cluster labels when one data point is replaced by another. To understand this, recall that the joint distribution of the labels $Y_1, \ldots, Y_{n - 1}$ is the product of $\mathrm{Cat}(\hat{\gamma}_n(X_1)), \ldots, \mathrm{Cat}(\hat{\gamma}_n(X_{n - 1}))$. If we replace the last data point $X_n$ with another point $X_{n + 1}$, the clustering results change, and the joint distribution of $Y_1, \ldots, Y_{n - 1}$ is the product of $\mathrm{Cat}(\hat{\gamma}_{n \to n + 1}(X_1)), \ldots, \mathrm{Cat}(\hat{\gamma}_{n \to n + 1}(X_{n - 1}))$. The total variation distance between these two joint distributions of the labels is essentially the $n$-th summand of $\mathsf{S}_n(\gamma)$. Conceptually, a small $\mathsf{S}_n$ means that a single input data point has a low influence on the clustering result, which is expected, especially when the sample size $n$ is large. Asymptotic replace-one invariance means that the influence of a single data point becomes negligible as the number of input points grows to infinity, and we will analyze certain cases where this holds in the next section.

Given these definitions, the following theorem provides a finite-sample bound on the coverage based on consistency and stability. 

\begin{theorem}
    \label{thm:coverage_bound}
    Suppose $X_1, \ldots, X_n \in \R^p$ are i.i.d.\ from $P_X^\ast$. Let $\hat{\cC}$ be the output of Algorithm \ref{alg:split_conformal_clustering_stochastic} with $X_1, \ldots, X_n$ as input, where we use stochastic clustering $\gamma$ and randomly split the data with $|\cI_{tr}| = |\cI_{ca}| = \frac{n}{2}$. Then, for $(X_{n + 1}, Y_{n + 1}^\ast) \sim P^\ast$ independent of $\{X_i\}_{i = 1}^{n}$, we have
    \begin{equation}
        \label{eq:coverage_bound_split}
        \P\left(\hat{\sigma}_o^\ast(Y_{n + 1}^\ast) \in \hat{\cC}(X_{n + 1})\right) \ge 1 - \alpha - \frac{n}{n + 2} \mathsf{E}_{n / 2}(\gamma) - \frac{n}{2 (n + 2)} \mathsf{S}_{n / 2}(\gamma).
    \end{equation}
\end{theorem}

Theorem~\ref{thm:coverage_bound} quantifies the under-coverage of our Conformal Clustering approach via two terms: consistency and stability. If these terms are small, the cluster-labeled data produced in Algorithm \ref{alg:split_conformal_clustering_stochastic} are comparable to samples from $P^\ast$, which leads to the desired $1 - \alpha$ coverage. This result also highlights the two key clustering ingredients that will ensure valid finite-sample coverage, thus providing a roadmap for future investigations.

\begin{remark}[Pitfall of Naive Split Conformal Clustering]
    \label{rmk:naive_split_conformal_clustering_undercoverage}
    While the main focus of Theorem \ref{thm:coverage_bound} is on Algorithm \ref{alg:split_conformal_clustering_stochastic}, which requires a stochastic clustering algorithm, the coverage bound \eqref{eq:coverage_bound_split} actually applies to any clustering algorithm and thus to naive Algorithm \ref{alg:split_conformal_clustering} as well. This is because any clustering algorithm can be represented as stochastic clustering via one-hot encoding. Any clustering algorithm that deterministically labels $x_1, \ldots, x_n$ with $y_1, \ldots, y_n \in [K]$ can be represented as $y_i \sim \mathrm{Cat}(\hat{\gamma}_n(x_i))$ with $\hat{\gamma}_n(x_i)$ whose $k$-th entry is $1$ if $k = y_i$ and $0$ otherwise. However, clustering algorithms based on such one-hot encoding tend to have larger $\mathsf{S}_n$ (less stable) unless the underlying cluster structure has extreme separation because $\gamma$ comprises one-hot encoding vectors. Similarly, the estimation error $\mathsf{E}_n$ can also be large when the underlying cluster structure is not well separated, namely, $\gamma^\ast(x)$ stays away from the vertices of the simplex $\Delta_K$ for most $x$'s, which is common in practice. As already noted in the previous sections, without stochastic clustering, Algorithm \ref{alg:split_conformal_clustering} generally fails in terms of both consistency and stability, which leads to severe under-coverage.
\end{remark}

\subsection{Asymptotic Coverage}
This section establishes asymptotic coverage results, where the right-hand side of \eqref{eq:coverage_bound_split} converges to $1 - \alpha$ as $n \to \infty$. First, the following theorem characterizes sufficient conditions under which the stochastic clustering $\gamma$ is consistent and asymptotically replace-one invariant. 

\begin{theorem}
    \label{thm:asymptotic_coverage}
    Under the setting of Theorem \ref{thm:coverage_bound}, suppose $\gamma$ is invariant to permutations of the input data, and the following conditions hold:
    \begin{align}
        \hat{\gamma}_n(X_1) - \gamma^\ast(X_1) & \xrightarrow{p} 0, \label{eq:marginal_consistency} \\
        H^2(\hat{\gamma}_n(X_2), \hat{\gamma}_{1 \to n + 1}(X_2)) & = o_p(n^{-1}), \label{eq:marginal_stability}
    \end{align}
    where $H^2(u, v) = \frac{1}{2} \sum_{k = 1}^{K} (\sqrt{u_k} - \sqrt{v_k})^2$ denotes the square of the Hellinger distance between any $u, v \in \Delta_K$. Then, the stochastic clustering $\gamma$ is consistent and asymptotically replace-one invariant, and we have $\P\left(\hat{\sigma}_o^\ast(Y_{n + 1}^\ast) \in \hat{\cC}(X_{n + 1})\right) \ge 1 - \alpha - o(1)$, where $o(1) \to 0$ as $n \to \infty$.
\end{theorem}

To summarize, Theorem \ref{thm:asymptotic_coverage} states that \eqref{eq:marginal_consistency} and \eqref{eq:marginal_stability} imply consistency and asymptotic replace-one invariance, respectively. Since $\mathsf{E}_n(\gamma)$ is the average estimation error over all input points, one can deduce that \eqref{eq:marginal_consistency}, the convergence for the first input $X_1$, will imply the convergence of the average $\mathsf{E}_n(\gamma)$. Meanwhile, it turns out that $\mathsf{S}_n(\gamma)$, the stability of the joint distribution of the labels, converges when \eqref{eq:marginal_stability}, the stability bound for a single input $X_2$, enjoys a sufficiently fast rate of $o_p(n^{-1})$.

Next, we show that these conditions are met by an important family of clustering algorithms: parametric mixture models, which are widely used in practice.

\begin{theorem}
    \label{thm:parametric_mixture_coverage}
    Let $\{P^\theta : \theta \in \Theta\}$ be a family of parametric mixture models for the joint distribution of $(X, Y^\ast)$ such that for $(X, Y^\ast) \sim P^\theta$, the conditional probability vector of $Y^\ast$ given $X$ can be represented as $(\P(Y^\ast = 1 \, | \, X = x), \ldots, \P(Y^\ast = K \, | \, X = x)) =: h(x, \theta)$, where $\theta \mapsto h(x, \theta)$ is a smooth function. Suppose $P^\ast = P^{\theta^\ast}$ for some $\theta^\ast$, and let $\hat{\theta}_n$ be a consistent root of the likelihood function given $X_1, \ldots, X_n \overset{i.i.d.}{\sim} P_X^\ast$ such that the standard asymptotic linearity holds, namely, 
    \begin{equation}
        \label{eq:asymptotic_linearity}
        \sqrt{n}(\hat{\theta}_n - \theta^\ast) = \frac{1}{\sqrt{n}} \sum_{i = 1}^{n} \psi(X_i) + O_p(n^{-1 / 2})
    \end{equation}
    for some $\psi \colon \R^p \to \Theta$ with $\E_{X \sim P_X^\ast} \psi(X) = 0$ and $\E_{X \sim P_X^\ast} \|\psi(X)\|_2^2 < \infty$. Suppose we implement the stochastic clustering $\gamma$ using the obtained estimator $\hat{\theta}_n$ by defining $\hat{\gamma}_n(x) = h(x, \hat{\theta}_n)$ for $x \in \R^p$, namely, the correctly specified parametric mixture model is used for clustering. Then, $\gamma$ is consistent and asymptotically replace-one invariant; thus we have $\P\left(\hat{\sigma}_o^\ast(Y_{n + 1}^\ast) \in \hat{\cC}(X_{n + 1})\right) \ge 1 - \alpha - o(1)$, where $o(1) \to 0$ as $n \to \infty$.
\end{theorem}

The key to Theorem \ref{thm:parametric_mixture_coverage} is that if the model parameter $\theta^\ast$ can be estimated consistently at the usual parametric rate,\footnote{It is common in the literature to have $o_p(1)$ instead of $O_p(n^{-1 / 2})$ on the right-hand side of \eqref{eq:asymptotic_linearity} because it is the minimum requirement needed to apply Slutsky's theorem and establish that $\sqrt{n}(\hat{\theta}_n - \theta^\ast)$ is asymptotically normal. However, in most cases under standard regularity conditions (such as the smoothness of $h$ in our case), the stronger $O_p(n^{-1 / 2})$ bound holds.} then $h(x, \hat{\theta}_n)$ will consistently approximate the true response probabilities. For smooth parametric mixture models, despite the existence of certain undesirable solutions to the likelihood equations, the classical result on the sequence of consistent roots applies under certain conditions, as shown in \citet{redner1984mixture}; see also \citet{wang2016maximum} for t mixture models. To find such roots, the EM algorithm is frequently employed, and is formally shown to converge, for instance, in Theorem 4.3 of \citet{redner1984mixture} under good initialization; see \citet{balakrishnan2017statistical} for a detailed and refined analysis.

In summary, Theorem \ref{thm:asymptotic_coverage} provides key conditions for consistency and asymptotic replace-one invariance, enabling the asymptotic coverage guarantee, while Theorem \ref{thm:parametric_mixture_coverage} exemplifies situations where these conditions hold by carefully deriving them from the existing results on the consistency of parametric mixture models. We conjecture that the conditions for asymptotic coverage in Theorem~\ref{thm:asymptotic_coverage} may hold more broadly for many other clustering methods beyond parametric mixture models. Currently, however, there is a lack of theory on the consistency of soft-label clustering methods beyond mixture models that can be exploited within our framework; this leaves ample room for future investigations. One particularly promising direction may be to consider nonparametric clustering algorithms that extend parametric mixture models through suitable transformations, such as mixture copula \citep{tewari2023estimation} or normalizing flows \citep{izmailov2020semi}. Verifying these conditions for wider classes of clustering algorithms is out of the scope of this paper, but we leave it as an important direction for future work.

\section{Empirical Studies}
\label{sec:simulations}
We empirically investigate our approach and validate our theoretical results through simulated examples in Section 4.1. We conclude with an application to cell type discovery from single-cell sequencing data in Section 4.2. Further simulations for additional clustering approaches and another application to astronomy data are given in the Supplementary Material. Data and \texttt{Python} code to reproduce all these simulations are provided at \url{https://github.com/DataSlingers/ConformalClustering}.

\subsection{Simulations on Parametric Mixture Models}
This section demonstrates the performance of our method (Algorithm \ref{alg:split_conformal_clustering_stochastic}) as well as validates our theoretical results on parametric mixture models, a Gaussian Mixture Model (GMM) and a Gamma Mixture Model (GaMM). We consider a low-dimensional setting ($p = 2$) with $K = 3$ components and a high-dimensional setting ($p = 50$ for GMM and $p = 30$ for GaMM) with $K = 5$ components, where the centers of the components are fixed and $\sigma^2 I_p$ is the common covariance matrix of the components.\footnote{Note that as clustering is more challenging in the higher-dimensional setting, we consider an adjusted range of $\sigma^2$ and configuration of the component centers to ensure that the clustering problem is not too easy or too difficult; details are given in the {\tt Python} code accompanying this paper.} We compare the performance of our method and the naive approaches for varying $\sigma^2$ with $n = 1000$ ($n = 5000$ for GMM and $n = 9000$ for GaMM in the higher-dimension) and for fixed $\sigma^2$ with varying $n$ to validate our asymptotic results in Section \ref{sec:theory}. Throughout the simulation, we always split the data into two halves (Step 1), use the support vector classifier with a radial basis function kernel for the classifier (Step 3), and use the generalized inverse quantile score defined in Section \ref{sec:background_conformal_prediction} for the conformity score (Step 4) in Algorithm \ref{alg:split_conformal_clustering_stochastic}, with $\alpha = 0.1$. While our theory currently only covers parametric mixture models, we conjecture that correct asymptotic coverage will hold more broadly based on Theorem~\ref{thm:asymptotic_coverage}. Hence, we provide empirical evidence for this using a stochastic Fuzzy-C-Means method in the Supplementary Material.

\begin{figure}[!tp]
    \centering
    \includegraphics[width=0.5\textwidth]{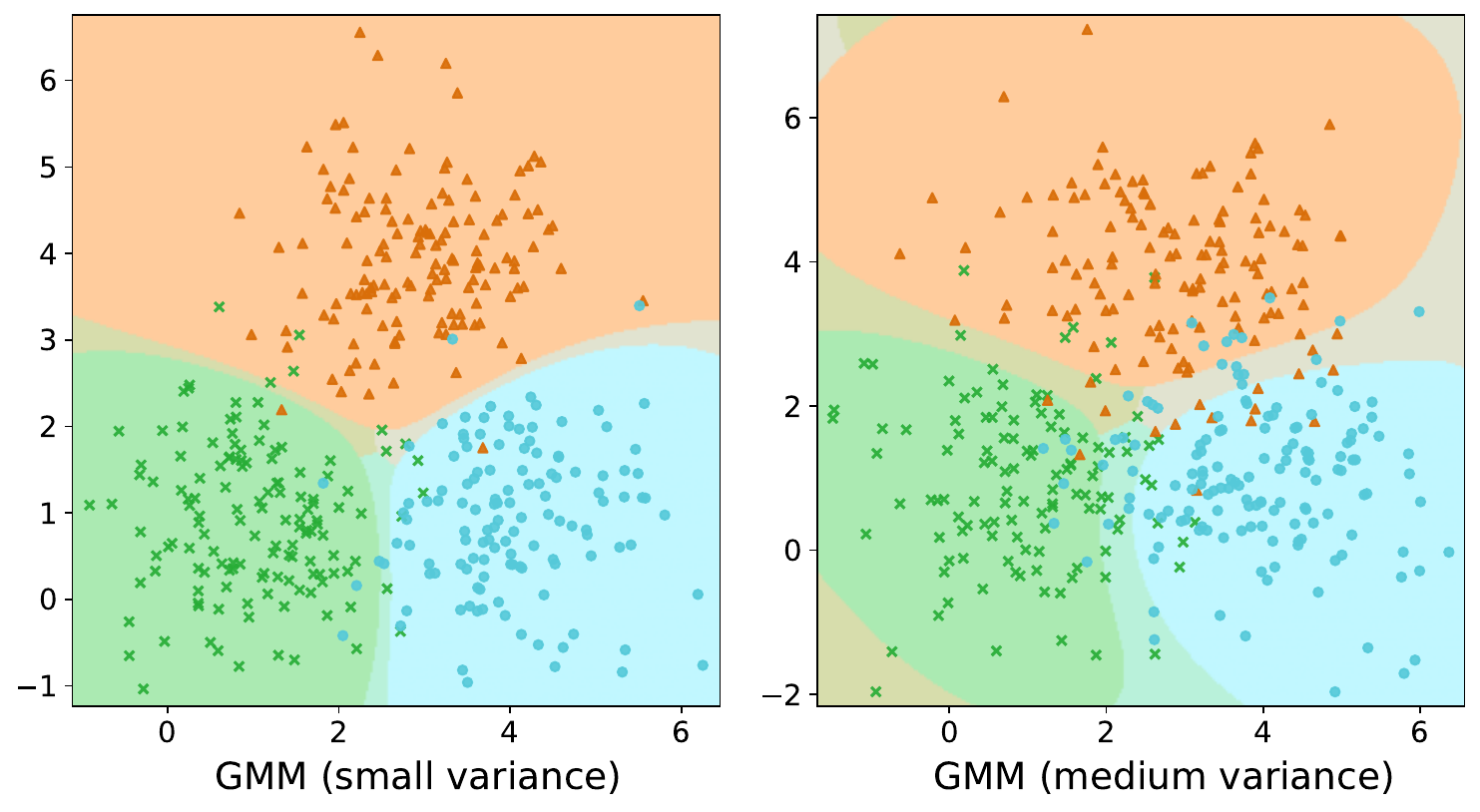}
    \hspace*{-7pt}
    \includegraphics[width=0.5\textwidth]{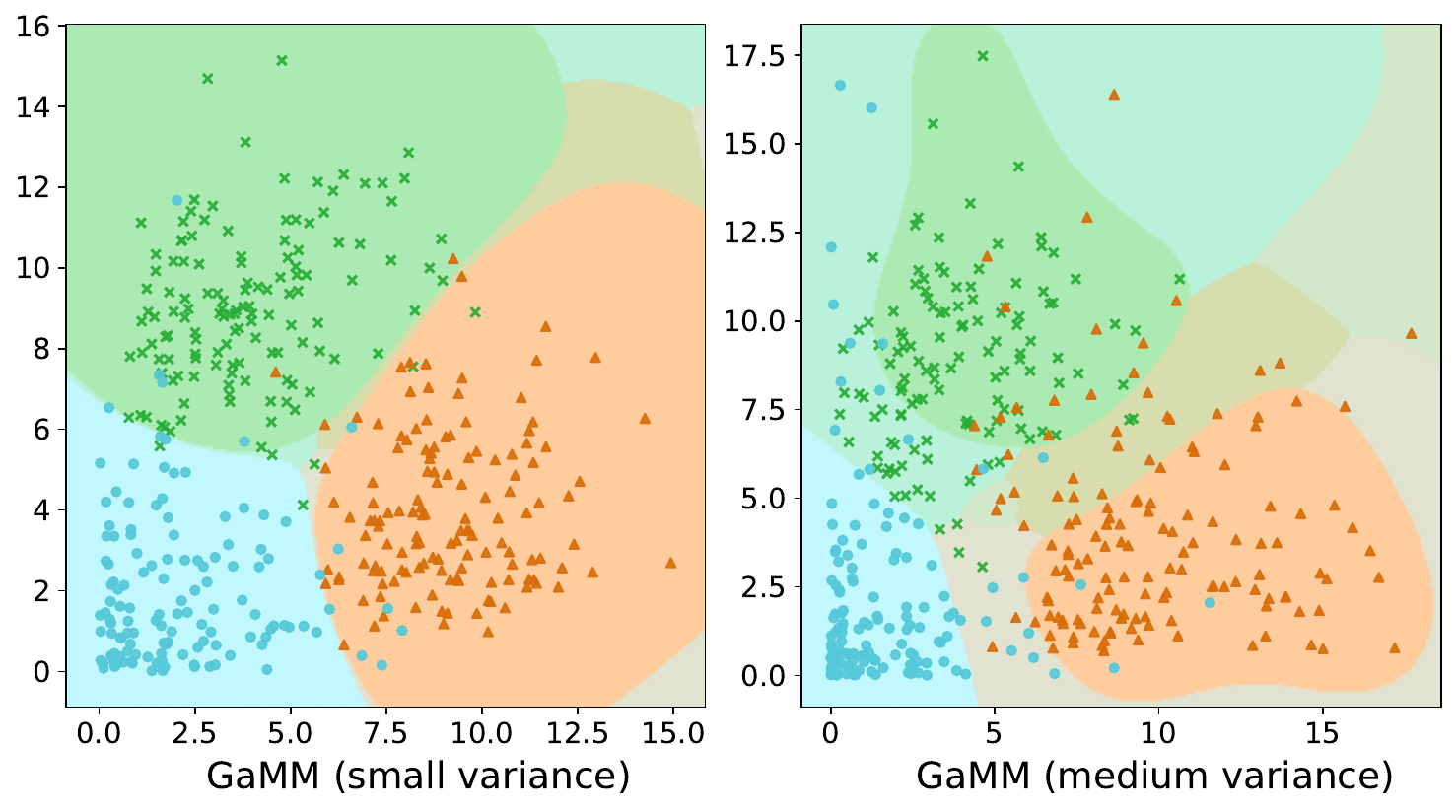}
    \caption{Confidence set heatmap visualization for GMM and GaMM simulations with three components in $\R^2$, with the sample ($n = 400$) colored with the true cluster labels. $\hat{\cC}(x)$ from Algorithm \ref{alg:split_conformal_clustering_stochastic}---with the corresponding stochastic mixture clustering---is visualized for a grid of covariate values $x$ in $\R^2$ by mixing the colors of the cluster labels in $\hat{\cC}(x)$. For both GMMs and GaMMs, $\sigma^2 I_2$ is the common covariance matrix of the three components, and we vary $\sigma^2$ to represent different levels of difficulty for the clustering problem.}
    \label{fig:heatmaps}
\end{figure}

\begin{figure}[!htb]
    \centering
    \includegraphics[width=\textwidth]{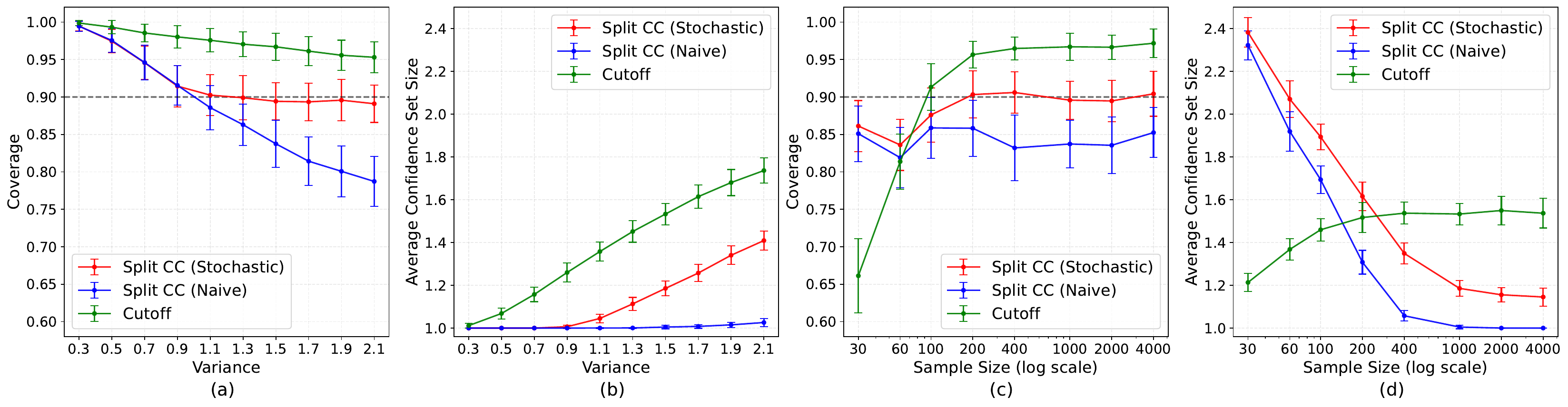}
    \includegraphics[width=\textwidth]{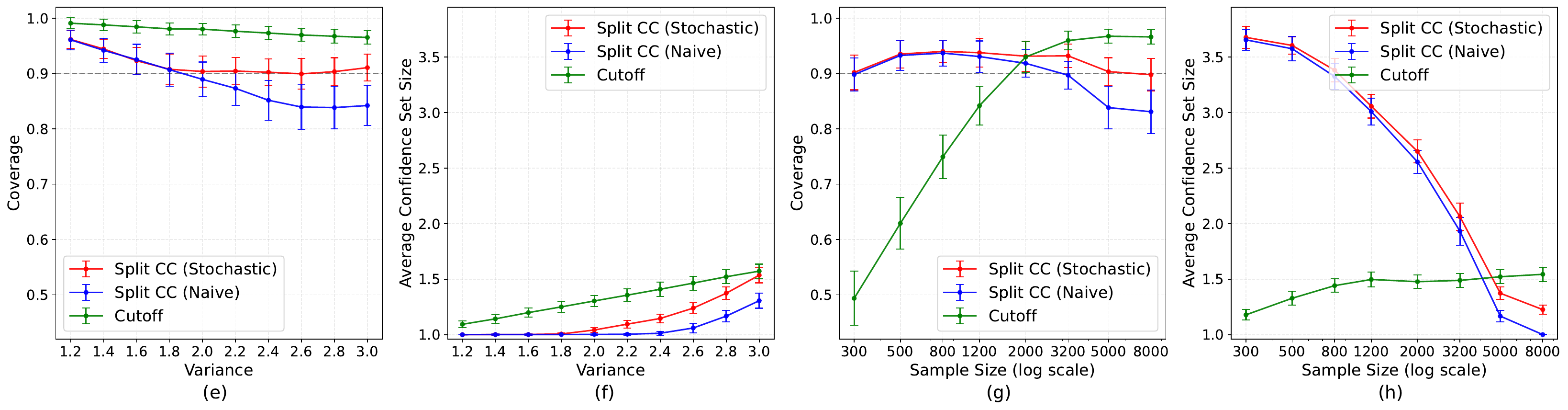}
    \caption{Coverage and confidence set size for GMM simulations with $p = 2$ and $K = 3$ (top row) and with $p = 50$ and $K = 5$ (bottom row). (a) and (b) show the results for varying $\sigma^2$ for fixed sample size $n = 1000$, while (c) and (d) are for fixed variance $\sigma^2 = 1.5$ and increasing sample size $n$. (e) and (f) are the results for varying $\sigma^2$ for fixed sample size $n = 5000$, while (g) and (h) are for fixed variance $\sigma^2 = 2.8$ and increasing sample size $n$. Our approach (Algorithm \ref{alg:split_conformal_clustering_stochastic}) is labeled as Split CC (Stochastic), while Algorithm \ref{alg:split_conformal_clustering} with hard clustering and Algorithm \ref{alg:cutoff} with thresholding soft labels are labeled as Split CC (Naive) and Cutoff, respectively.}
    \label{fig:GMM}
\end{figure}

To begin with, we first visualize the confidence sets produced by Algorithm \ref{alg:split_conformal_clustering_stochastic} in the two-dimensional setting for two different values of $\sigma^2$. The first two panels of Figure \ref{fig:heatmaps} show the points simulated from the above GMM, overlaid with the heatmap representing the confidence sets produced by Algorithm \ref{alg:split_conformal_clustering_stochastic} implemented with stochastic GMM clustering. The confidence set $\hat{\cC}(x) \subset \{1, 2, 3\}$ is visualized for a grid of covariate values $x \in \R^2$ by first assigning three colors (orange, blue, and green) to the cluster labels and coloring $x$ by a mixture of the colors of the labels in $\hat{\cC}(x)$. When $\sigma^2$ is small, the three components are well-separated, so the confidence sets are mostly singletons, with small regions corresponding to more than one cluster label around the cluster boundaries. For the larger $\sigma^2$, we observe wider regions with $|\hat{\cC}(x)| > 1$, reflecting the increased uncertainty in cluster label estimation. The last two panels of Figure \ref{fig:heatmaps} show the GaMM simulation results, which are qualitatively similar to the GMM simulation, with distinctions in the shapes of the clusters due to heavier tails and one-sided support of the gamma distribution. Overall, Figure \ref{fig:heatmaps} illustrates that successful employment of Algorithm \ref{alg:split_conformal_clustering_stochastic} can yield informative confidence sets that reflect the uncertainty in cluster label estimation.

We now evaluate the coverage and confidence set size of Algorithm \ref{alg:split_conformal_clustering_stochastic} on the GMM simulation. The top row of Figure \ref{fig:GMM} corresponds to the two-dimensional setting: (a) and (b) plot the coverage and set size for varying $\sigma^2$ and fixed sample size $n = 1000$, while (c) and (d) are for fixed variance $\sigma^2 = 1.5$ and increasing sample size $n$. From (a), we can see that our method achieves the desired $90\%$ coverage at the considered variance levels, while the naive split conformal clustering (Algorithm \ref{alg:split_conformal_clustering} with hard GMM clustering) experiences under-coverage. (b) shows that the confidence sets produced by the naive method are mostly singletons, failing to reflect the uncertainty in cluster label estimation. From (c) and (d), we can see that our method achieves the desired coverage for sufficiently large $n$, thus validating our asymptotic theory, and the confidence set size is informative as it gets smaller with increasing $n$. Crucially, observe that Algorithm \ref{alg:split_conformal_clustering} with hard clustering fails to achieve the desired coverage even as $n$ increases, which is because the hard clustering method is not stable and has large estimation error, thereby leading to under-coverage, as noted in Remark \ref{rmk:naive_split_conformal_clustering_undercoverage}. Meanwhile, we notice that Algorithm \ref{alg:cutoff} (Cutoff) is overly conservative, leading to unnecessarily large and thus uninformative confidence sets, confirming our discussion in Section \ref{sec:method}: despite the strong signal in the largest soft label, the cutoff method often has to include other labels (more than necessary) to exceed the $90\%$ threshold. The bottom row of Figure \ref{fig:GMM} shows the GMM results for the higher-dimensional setting. The results are qualitatively similar, with our method achieving the desired coverage and informative confidence set size, while the naive methods fail to achieve the desired coverage and/or produce uninformative confidence sets. Next, Figure \ref{fig:GaMM} shows the results for the GaMM simulation. Overall, the results reveal similar findings to the GMM simulation, once again validating our theoretical results and demonstrating the performance of our method in a different parametric mixture model setting.

\begin{figure}[!tb]
    \centering
    \includegraphics[width=\textwidth]{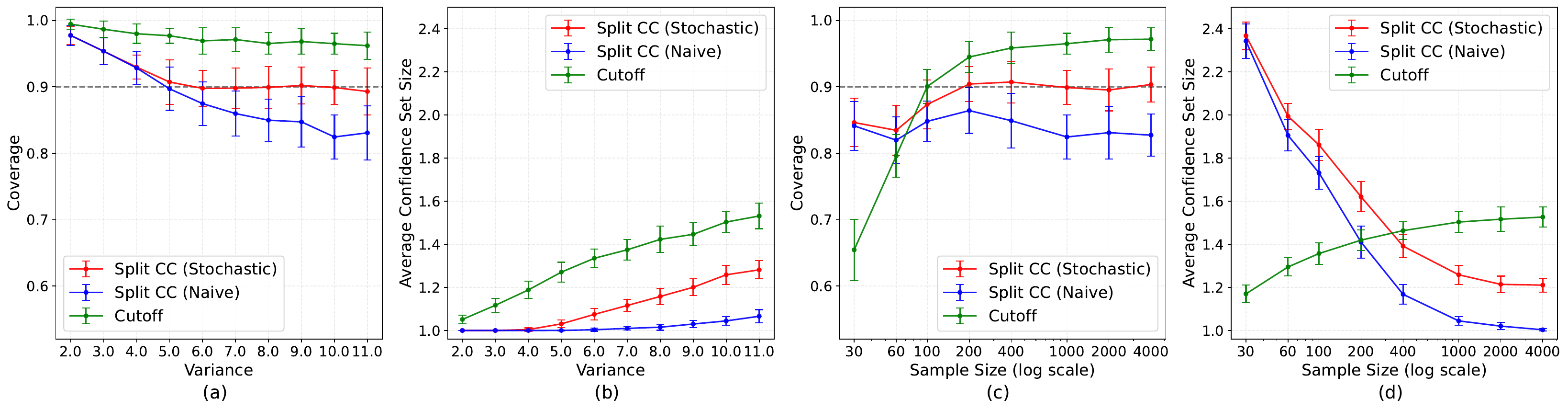}
    \includegraphics[width=\textwidth]{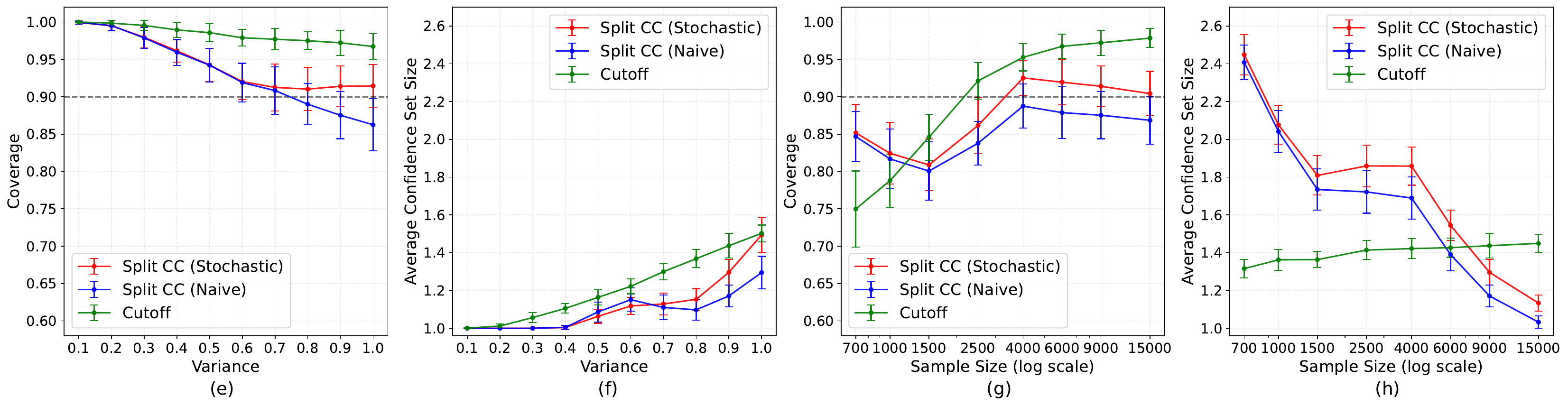}
    \caption{Coverage and confidence set size for GaMM simulations with $p = 2$ and $K = 3$ (top row) and with $p = 30$ and $K = 5$ (bottom row). (a) and (b) show the results for varying $\sigma^2$ for fixed sample size $n = 1000$, while (c) and (d) are for fixed variance $\sigma^2 = 10$ and increasing sample size $n$. (e) and (f) are the results for varying $\sigma^2$ for fixed sample size $n = 9000$, while (g) and (h) are for fixed variance $\sigma^2 = 0.9$ and increasing sample size $n$. The labeling of the methods is the same as in Figure \ref{fig:GMM}.}
    \label{fig:GaMM}
\end{figure}

In summary, the simulation results in this section demonstrate that our method, Algorithm \ref{alg:split_conformal_clustering_stochastic}, can successfully achieve the desired coverage and produce informative confidence sets. These validate our theoretical results in Section \ref{sec:theory}, while confirming our previous discussions on why the naive approaches fail to reflect the uncertainty in an informative manner.

\subsection{Application to Single-Cell Genomics}
\label{sec:applications}
We apply our Conformal Clustering approach to quantify the uncertainty in cell types discovered by clustering single-cell RNA sequencing data. Specifically, we analyze the PBMC-3K single-cell RNA-sequencing dataset,\footnote{Available at \url{https://www.10xgenomics.com/datasets/3-k-pbm-cs-from-a-healthy-donor-1-standard-1-1-0}.} a standard benchmark comprising approximately 2,700 peripheral blood mononuclear cells from a healthy donor. Following the standard pre-processing workflow of {\tt Scanpy} \citep{wolf2018scanpy} and {\tt Seurat} \citep{hao2024dictionary}, widely used software tools for single-cell data, we select the top 2,000 highly variable genes, normalize and log-transform expression counts, and obtain the reported cell-type labels through standard marker-gene matching available in these software packages. This pre-processing pipeline identifies $8$ cell-types, but we merge two of them, in accordance with the hierarchy of blood cell development, because one was only represented by a very small number of cells; this leaves us with $K = 7$ cell types. Note that the cell type labels are used solely as reference annotations for validation and visualization of our results.

\begin{figure}[!ht]
    \centering
    \subfloat{\includegraphics[width=\linewidth]{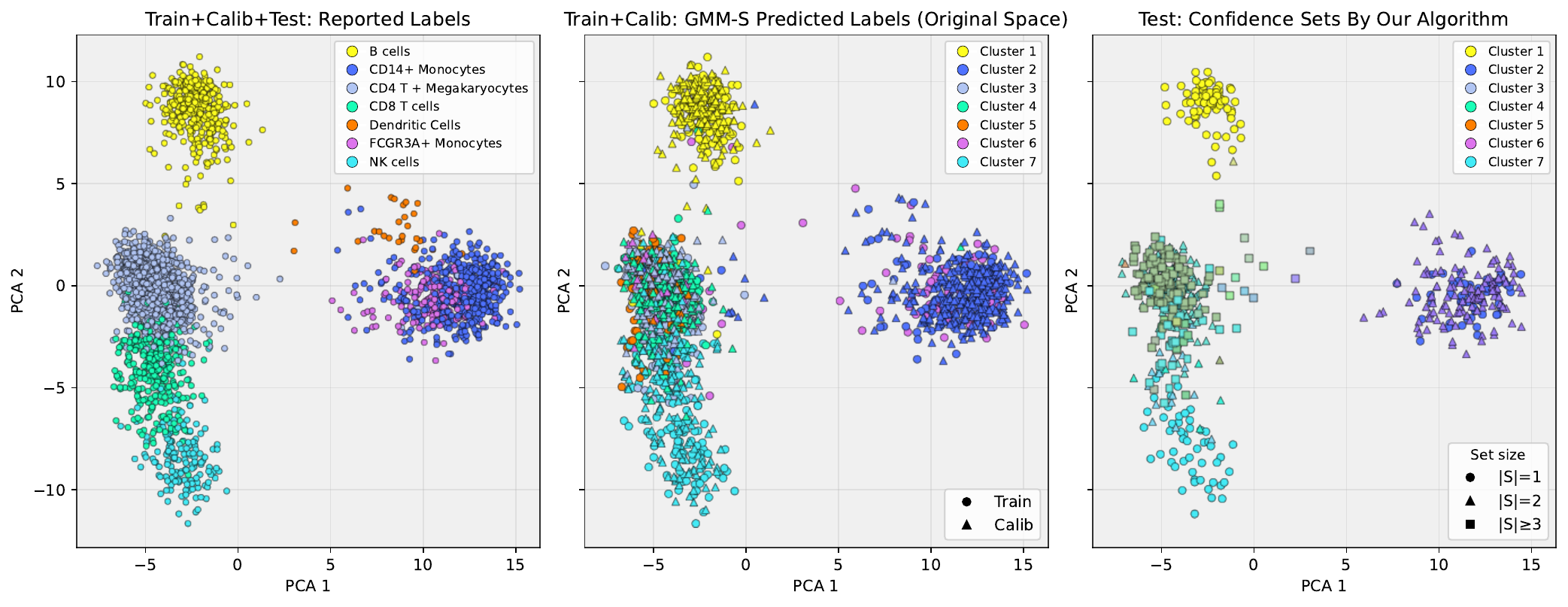}} \\
    \subfloat{\includegraphics[width=0.67\linewidth]{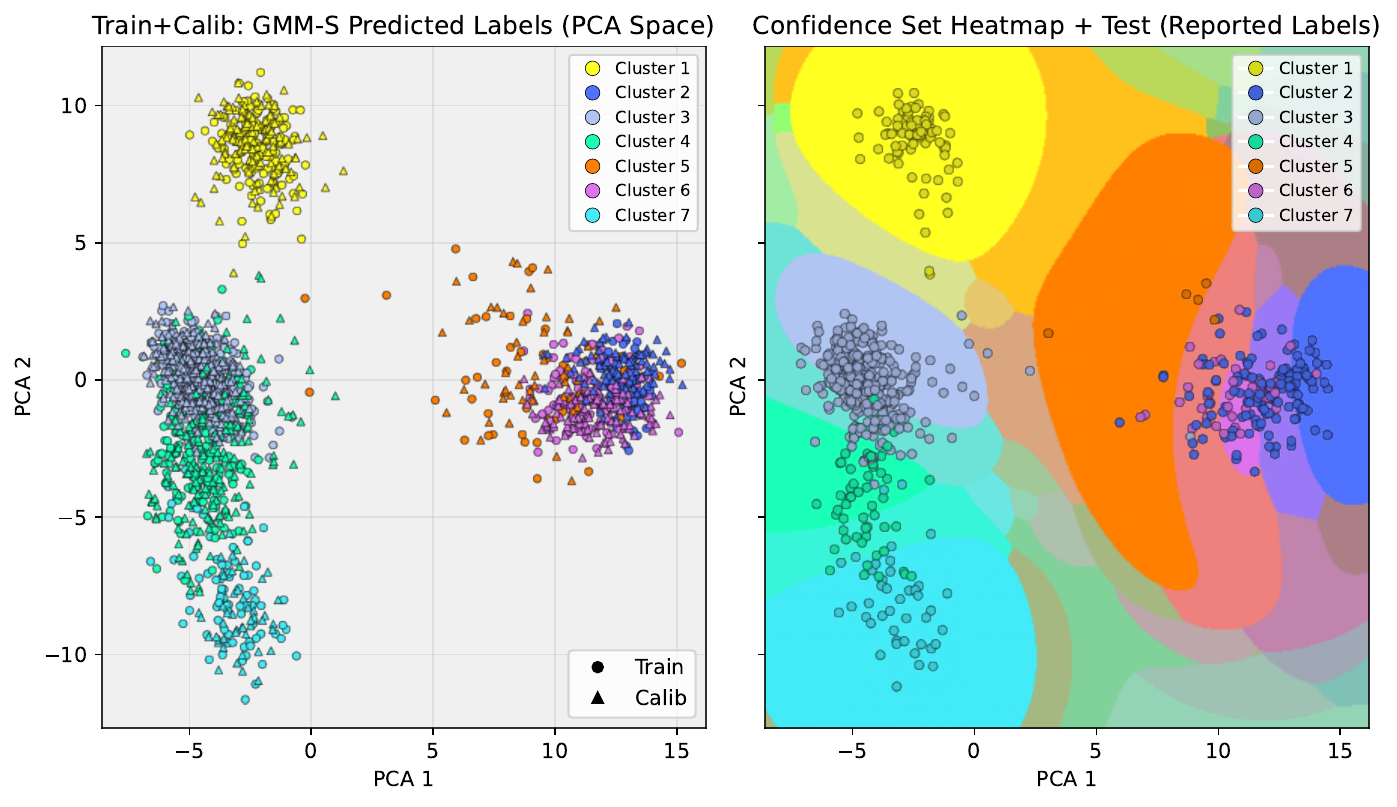}}
    \caption{Application of Algorithm \ref{alg:split_conformal_clustering_stochastic} to the PBMC single-cell RNA-seq dataset with stochastic GMM clustering. Top row: our method is applied in the original feature space and the outputs are visualized in the two-dimensional PCA space. From left to right, the plots show the reported reference labels, the predicted cluster labels, and the confidence sets for the test points. In the right most plot, marker shape encodes the set size $|\hat{\cC}(x)|$: circles for $|\hat{\cC}(x)| = 1$, triangles for $|\hat{\cC}(x)| = 2$, and squares for $|\hat{\cC}(x)| \ge 3$. Bottom row: our method is applied directly in the two-dimensional PCA space. The left plot shows the predicted cluster labels, and the right plot shows a heatmap of the confidence sets over the projected space with the test points overlaid using their reported labels. In both confidence-set visualizations, colors indicate the cluster labels contained in the confidence set.}
    \label{fig:pbmc_plots}
\end{figure}

We apply our Conformal Clustering algorithm in the original feature space using stochastic GMM clustering with $K = 7$ clusters, paired with a random forest classifier. We randomly sample $25\%$ of the data as the test set and equally split the remaining data into training and calibration sets. Given the high dimensionality ($p = 2000$), we restrict the GMM to diagonal covariance matrices to maintain computational feasibility.

We compute the confidence sets in the original feature space and visualize the results via a two-dimensional PCA projection (Figure \ref{fig:pbmc_plots}, top panel). The left plot displays the original observations colored by reference labels, the middle plot shows the cluster assignments predicted by stochastic GMM clustering for the training and calibration sets (Algorithm \ref{alg:split_conformal_clustering_stochastic}, Steps 2 and 4), and the right plot presents the resulting confidence sets for the test points, plotted in their PCA coordinates. Marker shapes indicate confidence set size: circles for singletons, triangles for pairs, and squares for sets of size $\ge 3$. We assign each cluster a base color (see legend) and color each confidence set by averaging the base colors of its contained labels. Lighter transparency reflects larger set sizes, providing a direct visual summary of cluster assignment uncertainty.

The results show B and NK cells exhibit high certainty, predominantly yielding singleton confidence sets. This is expected, as they differentiate into transcriptionally distinct populations with clear marker genes (e.g., CD79A, MS4A1, CD74 for B-cells; NKG7, GNLY, PRF1, FCGR3A for NK-cells), occupying well-separated transcriptomic regions \citep{Cai2020NK,Stubbington2017ImmuneSingleCell,schafflick2020integrated}. Conversely, greater uncertainty for CD4/CD8 T cells, monocytes, and dendritic cells reflects their underlying biology. T cells share a broad transcriptional backbone making fine-resolution annotation challenging, while monocytes and dendritic cells belong to a heterogeneous myeloid compartment with overlapping states \citep{Mullan2023TcellAnnotation,villani2017single}. Consequently, these groups naturally yield larger prediction sets than the more distinctive B and NK cells.

Although our method operates in any feature space, practitioners often prefer analyzing data in a lower-dimensional representation. Thus, we also illustrate our procedure applied directly within a PCA-projected space, again using stochastic GMM clustering ($K = 7$) and a support vector classifier. The results are displayed in the bottom panel of Figure \ref{fig:pbmc_plots}. The left plot shows the predicted cluster assignments for the training and calibration observations in the PCA space. The right plot presents a confidence-set heatmap computed over a 2D grid, colored by blending the contained labels' base colors, with overlaid test observations providing a visual anchor. These visualizations are informative, but they should be interpreted with care, since not every location in the PCA plane corresponds to a meaningful or well-supported region of the original high-dimensional space. As a result, apparent separation or confidence in the projection may sometimes be an artifact of extrapolation rather than a reflection of genuine structure. Therefore, while projected visualizations are illustrative, scientifically robust conclusions must remain anchored to the original feature space (Figure~\ref{fig:pbmc_plots} top panel). This application illustrates the potential of our conformal clustering framework as a scientifically useful downstream tool for cell-type identification while quantifying uncertainty in the cell type labels.

\section{Discussion}
\label{sec:discussion}

This paper introduces a principled framework for uncertainty quantification in clustering by constructing confidence sets for cluster labels that generalize across the feature space. We formally define the inference target and tackle the problem using a distribution-free inference perspective, making the methodology broadly applicable across different stochastic clustering algorithms. This work represents one of the first successful applications of conformal prediction to a purely unsupervised learning problem, paving the way for new methodological developments. We show that naively applying standard conformal classification techniques to predicted hard cluster labels fails to provide valid coverage guarantees, and that intuitive soft label cutoff-based methods similarly break down because they do not adapt to label uncertainty. To address this, we introduce stochastic cluster labels and a new framework for split conformal clustering; we demonstrate, both theoretically and empirically, that our approach provides a principled way to represent clustering uncertainty. Finally, we characterize the impact of estimated labels on coverage through the consistency and stability of the clustering algorithm, yielding an explicit account of the under-coverage caused by non-exchangeability and an asymptotic target coverage guarantee under mild structural assumptions. 

Several promising directions remain for future work. First, our theoretical guarantees currently assume that the number of clusters $K$ is known and fixed. In practice, $K$ is often chosen via data-driven heuristics, which introduces an additional layer of variability. Extending this framework to account for the post-selection inference of choosing $K$, or developing a conformal procedure that simultaneously selects $K$ and quantifies label uncertainty, would be highly valuable. Furthermore, since the labels generated by the clustering algorithm are stochastic, another interesting direction would be aggregating multiple conformal predictors based on different realizations or using cross-conformal-type aggregation for better statistical efficiency. Also, one could consider extending our results to cluster-conditional coverage, analogous to class-conditional coverage in \citet{sesia2025adaptivejrssb}.

Second, while our methodology fundamentally relies on stochastic labels, practitioners often favor hard-label clustering methods like K-means or hierarchical clustering. A potential bridge is to induce stochasticity via data perturbation, such as bootstrapping or subsampling, to generate an ensemble of cluster label assignments. This empirical smoothing process could adapt hard-label methods into the proposed framework while satisfying the necessary stability conditions. This would be an important direction to add to the growing literature for algorithmic stability \citep{soloff2024bagging, liang2025algorithmic}. Moreover, because the reliance on estimated labels naturally induces a form of distribution shift between the empirical cluster assignments and the latent ground truth, incorporating weighted conformal inference techniques \citep{tibshirani2019conformal} might explicitly correct for this discrepancy, potentially yielding tighter and more efficient confidence sets in finite samples.

In conclusion, by bridging the gap between unsupervised learning and distribution-free inference, our framework provides an actionable and rigorous tool for practitioners that quantifies uncertainty in cluster labels. By identifying which data points can be confidently assigned to a single cluster and which remain inherently ambiguous, this methodology substantially improves the interpretability, reliability, and trustworthiness of clustering analyses in scientific and industrial applications.

%%%%%%%%%%%%%%%%%%%%%%
%do not include in blinded version
\section*{Acknowledgments}
The authors acknowledge funding from NSF DMS-2516872.

\putbib
\end{bibunit}

\newpage
\pagenumbering{arabic}% resets `page` counter to 1
\renewcommand*{\thepage}{A\arabic{page}}
\begin{center}

\textbf{\large SUPPLEMENTARY MATERIAL}

\end{center}

\appendix
\begin{bibunit}
\section{Proofs}
\label{sec:proofs}

\subsection{Proof of Theorem \ref{thm:coverage_bound}}

To begin with, we recall the following result---originally from \citet{barber2023conformal}---on the coverage of nonexchangeable conformal inference procedures. For notational convenience, we consider a pretrained version of split conformal classification, where the soft classifier is assumed to be obtained independently of the input data. The following Theorem is adapted from Theorem 7.12 of \citet{angelopoulos2024theoretical}.

\begin{theorem}
    \label{thm:conformal_bound_nonexchangeable}
    Suppose $(X_1, Y_1^\ast), \ldots, (X_{n + 1}, Y_{n + 1}^\ast) \in \R^p \times [K]$ are random variables that are not necessarily i.i.d. Let $\hat{\cC}$ be the output of split conformal classification with $(X_1, Y_1^\ast), \ldots, (X_n, Y_n^\ast)$ as input and a fixed soft classifier that is trained independently of the data, which takes the following form: for any $x \in \R^p$, define
    \begin{equation*}
        \hat{\cC}(x) = \Big\{y \in [K] : s((x, y); \hat{\pi}) \le \lceil (1 - \alpha) (1 + |\cI_{ca}|)\rceil \text{-th smallest value of} ~ \{s_i\}_{i \in \cI_{ca}}\Big\},
    \end{equation*}
    where $s(\cdot; \hat{\pi}) \colon \R^p \times [K] \to \R$ is a fixed score function defined based on the soft classifier $\hat{\pi}$, and $s_i := s((X_i, Y_i^\ast); \hat{\pi}) \in \R$ for $i = 1, \ldots, n$. Then, we have
    \begin{equation*}
        \P\left(Y_{n + 1}^\ast \in \hat{\cC}(X_{n + 1})\right)
        \ge 1 - \alpha - \frac{1}{n + 1} \sum_{i = 1}^{n} \mathsf{TV}(\mu, \mu^{i \leftrightarrow n + 1}),  
    \end{equation*}
    where $\mu$ is the probability measure on $\R^{n + 1}$ corresponding to the joint distribution of the scores $s((X_1, Y_1^\ast); \hat{\pi}), \ldots, s((X_{n + 1}, Y_{n + 1}^\ast); \hat{\pi})$, while $\mu^{i \leftrightarrow n + 1}$ is the pushforward measure of $\mu$ by the map on $\R^{n + 1}$ that swaps the $i$-th and $(n + 1)$-th coordinates.
\end{theorem}

Now we prove Theorem \ref{thm:coverage_bound}. 

\begin{proof}[Proof of Theorem \ref{thm:coverage_bound}]
    As above, we consider a pretrained version of Algorithm \ref{alg:split_conformal_clustering_stochastic} where
    $\cI_{ca} = [n]$. We will concretely prove the following:
    \begin{equation}
        \label{eq:coverage_bound_pre}
        \P\left(\hat{\sigma}_o^\ast(Y_{n + 1}^\ast) \in \hat{\cC}(X_{n + 1})\right) \ge 1 - \alpha - \frac{n}{n + 1} \mathsf{E}_n(\gamma) - \frac{n}{2 (n + 1)} \mathsf{S}_n(\gamma).
    \end{equation}
    As we treat the soft classifier $\hat{\pi}$ as a fixed map, let $s(x, y)$ denote the score for $(x, y) \in \R^p \times [K]$, instead of $s((x, y); \hat{\pi})$. By Theorem \ref{thm:conformal_bound_nonexchangeable}, 
    \begin{equation*}
        \P\left(\hat{\sigma}_o^\ast(Y_{n + 1}^\ast) \in \hat{\cC}(X_{n + 1})\right) \ge 1 - \alpha - \frac{1}{n + 1} \sum_{i = 1}^{n} \mathsf{TV}(\mu, \mu^{i \leftrightarrow n + 1}),
    \end{equation*}
    where $\mu$ represents the joint distribution of $s(X_1, \hat{\sigma}(Y_1)), \ldots, s(X_n, \hat{\sigma}(Y_n)), s(X_{n + 1}, \hat{\sigma}_o^\ast(Y_{n + 1}^\ast))$, while $\mu^{i \leftrightarrow n + 1}$ is the pushforward of $\mu$ by swapping the $i$-th and $(n + 1)$-th coordinates. Note that these scores can be represented as a function of $(X_1, Y_1), \ldots, (X_n, Y_n), (X_{n + 1}, Y_{n + 1}^\ast)$, namely, 
    \begin{equation*}
        \left(s(X_1, \hat{\sigma}(Y_1)), \ldots, s(X_n, \hat{\sigma}(Y_n)), s(X_{n + 1}, \hat{\sigma}_o^\ast(Y_{n + 1}^\ast))\right) = S\left((X_1, Y_1), \ldots, (X_n, Y_n), (X_{n + 1}, Y_{n + 1}^\ast)\right)
    \end{equation*}
    for some deterministic function $S$. Let $\rho$ be the probability measure corresponding to the joint distribution of 
    \begin{equation*}
        (X_1, Y_1), \ldots, (X_n, Y_n), (X_{n + 1}, Y_{n + 1}^\ast).
    \end{equation*}
    Hence, $\mu$ is the pushforward measure of $\rho$ by the map $S$. Meanwhile, notice that $\mu^{i \leftrightarrow n + 1}$ corresponds to the joint distribution of 
    \begin{equation*}
        S\left((X_1, Y_1), \ldots, (X_{i - 1}, Y_{i - 1}), (X_{n + 1}, Y_{n + 1}^\ast), (X_{i + 1}, Y_{i + 1}), \ldots, (X_n, Y_n), (X_i, Y_i)\right).
    \end{equation*}
    Let $\rho^{i \leftrightarrow n + 1}$ be the probability measure corresponding to the joint distribution of
    \begin{equation*}
        (X_1, Y_1), \ldots, (X_{i - 1}, Y_{i - 1}), (X_{n + 1}, Y_{n + 1}^\ast), (X_{i + 1}, Y_{i + 1}), \ldots, (X_n, Y_n), (X_i, Y_i).
    \end{equation*}
    Then, $\mu^{i \leftrightarrow n + 1}$ is the pushforward measure of $\rho^{i \leftrightarrow n + 1}$ by $S$. By the data processing inequality, 
    \begin{equation*}
        \mathsf{TV}(\mu, \mu^{i \leftrightarrow n + 1}) \le \mathsf{TV}(\rho, \rho^{i \leftrightarrow n + 1}).
    \end{equation*}
    Hence,
    \begin{equation}
        \label{eq:coverage_bound_tmp}
        \P\left(\hat{\sigma}_o^\ast(Y_{n + 1}^\ast) \in \hat{\cC}(X_{n + 1})\right) \ge 1 - \alpha - \frac{1}{n + 1} \sum_{i = 1}^{n} \mathsf{TV}(\rho, \rho^{i \leftrightarrow n + 1}).
    \end{equation}
    We analyze $\mathsf{TV}(\rho, \rho^{i \leftrightarrow n + 1})$. For notational simplicity, let $i = n$, and analyze $\mathsf{TV}(\rho, \rho^{n \leftrightarrow n + 1})$. Essentially, we need to compare the total variation distance between the joint distributions of the following two sets of variables:
    \begin{align*}
        & (X_1, Y_1), \ldots, (X_{n - 1}, Y_{n - 1}), (X_n, Y_n), (X_{n + 1}, Y_{n + 1}^\ast), \\
        & (X_1, Y_1), \ldots, (X_{n - 1}, Y_{n - 1}), (X_{n + 1}, Y_{n + 1}^\ast), (X_n, Y_n).
    \end{align*}
    Notice that the joint distribution of $Y_1, \ldots, Y_n, Y_{n + 1}^\ast$ given $(X_1, \ldots, X_{n + 1}) = (x_1, \ldots, x_{n + 1})$ is 
    \begin{equation*}
        \Lambda(x_1, \ldots, x_{n + 1}) := \otimes_{j = 1}^{n - 1} \mathrm{Cat}(\hat{\gamma}_n(x_j)) \otimes \mathrm{Cat}(\hat{\gamma}_n(x_n)) \otimes \mathrm{Cat}(\gamma^\ast(x_{n + 1})),
    \end{equation*}
    where $\hat{\gamma}_n$ is fitted on $x_1, \ldots, x_n$. Meanwhile, the joint distribution of $Y_1, \ldots, Y_{n - 1}, Y_{n + 1}^\ast, Y_n$ given $(X_1, \ldots, X_{n - 1}, X_{n + 1}, X_n) = (x_1, \ldots, x_{n + 1})$ is 
    \begin{equation*}
        \Lambda'(x_1, \ldots, x_{n + 1}) := \otimes_{j = 1}^{n - 1} \mathrm{Cat}(\hat{\gamma}_{n \to n + 1}(x_j)) \otimes \mathrm{Cat}(\gamma^\ast(x_n)) \otimes \mathrm{Cat}(\hat{\gamma}_{n \to n + 1}(x_{n + 1})),
    \end{equation*}
    where $\hat{\gamma}_{n \to n + 1}$ is fitted on $x_1, \dots, x_{n - 1}, x_{n + 1}$. Since $(X_1, \ldots, X_{n + 1})$ and $(X_1, \ldots, X_{n - 1}, X_{n + 1}, X_n)$ have the same joint distribution $(P_X^\ast)^{\otimes (n + 1)}$, we have the following from Lemma \ref{lem:tv_conditional}:
    \begin{equation*}
        \mathsf{TV}(\rho, \rho^{n \leftrightarrow n + 1}) \le \int \mathsf{TV}\left(\Lambda(x_1, \ldots, x_{n + 1}), \Lambda'(x_1, \ldots, x_{n + 1})\right) \, \mathrm{d} \nu_{n + 1}(x_1, \ldots, x_{n + 1}),
    \end{equation*}
    where $\nu_m$ denotes $(P_X^\ast)^{\otimes m}$ for any $m \in \N$. As $\mathsf{TV}(P_1 \otimes P_2, Q_1 \otimes Q_2) \le \sum_{i = 1}^{2} \mathsf{TV}(P_i, Q_i)$, we have
    \begin{equation*}
        \begin{split}
            \mathsf{TV}(\rho, \rho^{n \leftrightarrow n + 1}) & \le \int \mathsf{TV}\left(\otimes_{j = 1}^{n - 1} \mathrm{Cat}(\hat{\gamma}_n(x_j)), \otimes_{j = 1}^{n - 1} \mathrm{Cat}(\hat{\gamma}_{n \to n + 1}(x_j))\right) \, \mathrm{d} \nu_{n + 1}(x_1, \ldots, x_{n + 1}) \\
            & \quad + \int \mathsf{TV}\left(\mathrm{Cat}(\hat{\gamma}_n(x_n)), \mathrm{Cat}(\gamma^\ast(x_n))\right) \, \mathrm{d} \nu_n(x_1, \ldots, x_n) \\
            & \quad + \int \mathsf{TV}\left(\mathrm{Cat}(\hat{\gamma}_{n \to n + 1}(x_{n + 1})), \mathrm{Cat}(\gamma^\ast(x_{n + 1}))\right) \, \mathrm{d} \nu_n(x_1, \ldots, x_{n - 1}, x_{n + 1}),
        \end{split}
    \end{equation*}
    where we note that the last two terms are equal due to the change of variables. Notice that for any $\gamma_1, \ldots, \gamma_m \in \Delta_K$ and $\gamma_1', \ldots, \gamma_m' \in \Delta_K$, we have
    \begin{equation*}
        \mathsf{TV}\left(\otimes_{j = 1}^{m} \mathrm{Cat}(\gamma_j), \otimes_{j = 1}^{m} \mathrm{Cat}(\gamma_j')\right) = \frac{1}{2} \Big\|\motimes_{j = 1}^{m} \gamma_j - \motimes_{j = 1}^{m} \gamma_j'\Big\|_1.
    \end{equation*}
    Hence,
    \begin{equation*}
        \begin{split}
            \mathsf{TV}(\rho, \rho^{n \leftrightarrow n + 1})
            & \le \frac{1}{2} \int \Big\|\motimes_{j = 1}^{n - 1} \hat{\gamma}_n(x_j) - \motimes_{j = 1}^{n - 1} \hat{\gamma}_{n \to n + 1}(x_j)\Big\|_1 \, \mathrm{d} \nu_{n + 1}(x_1, \ldots, x_{n + 1}) \\
            & \quad + \int \|\hat{\gamma}_n(x_n) - \gamma^\ast(x_n)\|_1 \, \mathrm{d} \nu_n(x_1, \ldots, x_n).
        \end{split}
    \end{equation*}
    From this, we deduce that 
    \begin{equation*}
        \sum_{i = 1}^{n} \mathsf{TV}(\rho, \rho^{i \leftrightarrow n + 1})
        \le n \mathsf{E}_n(\gamma) + \frac{n}{2} \mathsf{S}_n(\gamma).
    \end{equation*}
    Combining this with \eqref{eq:coverage_bound_tmp}, we obtain \eqref{eq:coverage_bound_pre}. Substituting $n / 2$ for $n$ to reflect the equal data splitting used in the algorithm yields the target bound \eqref{eq:coverage_bound_split}.
\end{proof}

\subsection{Proof of Theorem \ref{thm:asymptotic_coverage}}
\begin{proof}[Proof of Theorem \ref{thm:asymptotic_coverage}]
    Since $\|\hat{\gamma}_n(X_1) - \gamma^\ast(X_1)\|_1 \le 2$ by definition, $\|\hat{\gamma}_n(X_1) - \gamma^\ast(X_1)\|_1 \xrightarrow{p} 0$ implies the convergence in expectation $\E \|\hat{\gamma}_n(X_1) - \gamma^\ast(X_1)\|_1 \to 0$. Hence, by symmetry of $\gamma$, we have $\mathsf{E}_n(\gamma) = \E \|\hat{\gamma}_n(X_1) - \gamma^\ast(X_1)\|_1 \to 0$. 

    Similarly, by symmetry of $\gamma$, the stability term $\mathsf{S}_n(\gamma)$ equals the expectation of 
    \begin{equation*}
        \Xi_n := \Big\|\motimes_{j = 2}^{n} \hat{\gamma}_n(X_j) - \motimes_{j = 2}^{n} \hat{\gamma}_{1 \to n + 1}(X_j)\Big\|_1 = 2 \mathsf{TV}\left(\motimes_{j = 2}^{n} \hat{\gamma}_n(X_j), \motimes_{j = 2}^{n} \hat{\gamma}_{1 \to n + 1}(X_j)\right).
    \end{equation*}
    As $\Xi_n \le 2$, it suffices to show that $\Xi_n \xrightarrow{p} 0$. Again, by $\mathsf{TV} \le \sqrt{2 H^2}$ and the property of the Hellinger distance for product distributions, we have
    \begin{equation*}
        \begin{split}
            \Xi_n
            & \le 2 \sqrt{2 H^2\left(\motimes_{j = 2}^{n} \hat{\gamma}_n(X_j), \motimes_{j = 2}^{n} \hat{\gamma}_{1 \to n + 1}(X_j)\right)} \\
            & = 2 \sqrt{2 \left(1 - \prod_{j = 2}^{n} \left(1 - H^2\left(\hat{\gamma}_n(X_j), \hat{\gamma}_{1 \to n + 1}(X_j)\right)\right)\right)}.
        \end{split}
    \end{equation*}
    To prove $\Xi_n \xrightarrow{p} 0$, note that it suffices to show that
    \begin{equation*}
        \sum_{j = 2}^{n} \log\left(1 - H^2\left(\hat{\gamma}_n(X_j), \hat{\gamma}_{1 \to n + 1}(X_j)\right)\right) \xrightarrow{p} 0.
    \end{equation*}
    By Lemma \ref{lem:sum_convergence_to_log}, it suffices to show that 
    \begin{equation}
        \label{eq:hellinger_sum_convergence}
        \sum_{j = 2}^{n} H^2\left(\hat{\gamma}_n(X_j), \hat{\gamma}_{1 \to n + 1}(X_j)\right) \xrightarrow{p} 0.
    \end{equation}
    By the given condition \eqref{eq:marginal_stability} and symmetry of $\gamma$, we deduce that
    \begin{equation*}
        \sum_{j = 2}^{n} H^2\left(\hat{\gamma}_n(X_j), \hat{\gamma}_{1 \to n + 1}(X_j)\right) = (n - 1) o_p(n^{-1}) \xrightarrow{p} 0.
    \end{equation*}
    Hence, $\Xi_n \xrightarrow{p} 0$ and $\mathsf{S}_n(\gamma) = \E[\Xi_n] \to 0$.
\end{proof}

\subsection{Proof of Theorem \ref{thm:parametric_mixture_coverage}}
First, let us rewrite $h$ from Theorem \ref{thm:parametric_mixture_coverage} as follows:
\begin{equation*}
    h(x, \theta) = ([h(x, \theta)]_1, \ldots, [h(x, \theta)]_K) \quad \forall x \in \R^p.
\end{equation*}
We can think of this as a probability mass function $k \mapsto [h(x, \theta)]_k$. Then, the corresponding score function is a collection of maps
\begin{equation*}
    k \mapsto \frac{\partial \log [h(x, \theta)]_k}{\partial \theta_j} \quad \forall j,
\end{equation*}
which is well-defined as $h$ is smooth with respect to the parameter $\theta$.\footnote{One technicality here is that $\theta$ lies in some manifold. By reparametrizing the mixing proportions $w_1, \ldots, w_K$ and the covariance matrices $\Sigma_1, \ldots, \Sigma_K$, we can treat $\theta$ as a variable in some open subset $\Theta$ in some Euclidean space.} 

The following lemma establishes the usual identity of the score function and the Fisher information and the second-order Taylor expansion for the squared Hellinger distance.

\begin{lemma}
    \label{lem:conditional_density_distance_approx}
    For $h$ in Theorem \ref{thm:parametric_mixture_coverage}, the score function satisfies the following identity: for any $j$, we have
    \begin{equation}
        \label{eq:score_identity}
        \sum_{k = 1}^{K} [h(x, \theta)]_k \frac{\partial \log [h(x, \theta)]_k}{\partial \theta_j} = \sum_{k = 1}^{K} \frac{\partial [h(x, \theta)]_k}{\partial \theta_j} = 0,
    \end{equation}
    and also for any pair $(j, \ell)$, we have
    \begin{equation}
        \label{eq:score_identity2}
        \sum_{k = 1}^{K} \frac{\partial^2 [h(x, \theta)]_k}{\partial \theta_j \partial \theta_\ell} = 0.
    \end{equation}
    Accordingly, the Fisher information matrix is well-defined, which we denote as $I_x(\theta)$ whose $(j, \ell)$-th entry is defined by
    \begin{equation*}
        [I_x(\theta)]_{j \ell} := \sum_{k = 1}^{K} [h(x, \theta)]_k \frac{\partial \log [h(x, \theta)]_k}{\partial \theta_j} \frac{\partial \log [h(x, \theta)]_k}{\partial \theta_\ell} = \sum_{k = 1}^{K} \frac{1}{[h(x, \theta)]_k} \frac{\partial [h(x, \theta)]_k}{\partial \theta_j} \frac{\partial [h(x, \theta)]_k}{\partial \theta_\ell}.
    \end{equation*}
    Similarly, the following usual identity for the Fisher information matrix also holds: for every pair $(j, \ell)$, we have
    \begin{equation*}
        [I_x(\theta)]_{j \ell} = - \sum_{k = 1}^{K} [h(x, \theta)]_k \frac{\partial^2 \log [h(x, \theta)]_k}{\partial \theta_j \partial \theta_\ell}.        
    \end{equation*}
    Then, the squared Hellinger distance between two probability mass functions $k \mapsto [h(x, \theta_0)]_k$ and $k \mapsto [h(x, \theta)]_k$ satisfies the following:
    \begin{equation}
        \label{eq:hellinger_approx}
        H^2(h(x, \theta_0), h(x, \theta)) = \frac{1}{8} \langle \theta - \theta_0, I_x(\theta_0) (\theta - \theta_0) \rangle + o(\|\theta - \theta_0\|^2).
    \end{equation}
\end{lemma}
\begin{proof}
    Since $\sum_{k = 1}^{K} [h(x, \theta)]_k = 1$ by construction, we have the usual identity of the score function. This confirms \eqref{eq:score_identity} and \eqref{eq:score_identity2}. Similarly, one can check the usual identity for the Fisher information matrix. Now, let
    \begin{equation*}
        1 - H^2(h(x, \theta_0), h(x, \theta)) = \sum_{k = 1}^{K} \sqrt{[h(x, \theta_0)]_k [h(x, \theta)]_k} =: B(\theta).
    \end{equation*}
    Note that $B(\theta_0) = 1$ and
    \begin{equation*}
        [\nabla B(\theta)]_j = \sum_{k = 1}^{K} \frac{\sqrt{[h(x, \theta_0)]_k}}{2 \sqrt{[h(x, \theta)]_k}} \frac{\partial [h(x, \theta)]_k}{\partial \theta_j},
    \end{equation*}
    which implies that $\nabla B(\theta_0) = 0$ by \eqref{eq:score_identity}. Also, the Hessian of $B$ is given by
    \begin{equation*}
        [\nabla^2 B(\theta)]_{j \ell} = \sum_{k = 1}^{K} \frac{\sqrt{[h(x, \theta_0)]_k}}{2 \sqrt{[h(x, \theta)]_k}} \frac{\partial^2 [h(x, \theta)]_k}{\partial \theta_j \partial \theta_\ell} - \sum_{k = 1}^{K} \frac{\sqrt{[h(x, \theta_0)]_k}}{4 [h(x, \theta)]_k \sqrt{[h(x, \theta)]_k}} \frac{\partial [h(x, \theta)]_k}{\partial \theta_\ell} \frac{\partial [h(x, \theta)]_k}{\partial \theta_j}.
    \end{equation*}
    By \eqref{eq:score_identity2}, we have 
    \begin{equation*}
        [\nabla^2 B(\theta_0)]_{j \ell} = - \sum_{k = 1}^{K} \frac{1}{4 [h(x, \theta_0)]_k} \frac{\partial [h(x, \theta)]_k}{\partial \theta_\ell} \bigg \vert_{\theta = \theta_0} \frac{\partial [h(x, \theta)]_k}{\partial \theta_j} \bigg \vert_{\theta = \theta_0} = -\frac{1}{4} [I_x(\theta_0)]_{j \ell}.
    \end{equation*}
    Since $B$ is a smooth function, we have 
    \begin{equation*}
        B(\theta) = B(\theta_0) + \langle \nabla B(\theta_0), \theta - \theta_0 \rangle + \frac{1}{2} \langle \theta - \theta_0, \nabla^2 B(\theta_0) (\theta - \theta_0) \rangle + o(\|\theta - \theta_0\|^2).
    \end{equation*}
    Hence, we have \eqref{eq:hellinger_approx}.
\end{proof}

Now, we prove Theorem \ref{thm:parametric_mixture_coverage}.
\begin{proof}[Proof of Theorem \ref{thm:parametric_mixture_coverage}]
    Recall that the consistent root $\hat{\theta}_n$ satisfies
    \begin{equation*}
        \hat{\theta}_n \xrightarrow{p} \theta^\ast \quad \text{and} \quad \hat{\theta}_n - \theta^\ast = O_p(n^{-1 / 2}).
    \end{equation*}
    We claim that $\|\hat{\gamma}_n(X_1) - \gamma^\ast(X_1)\|_1 \xrightarrow{p} 0$. As $\gamma^\ast(X_1) = h(X_1, \theta^\ast)$ and $\hat{\gamma}_n(X_1) = h(X_1, \hat{\theta}_n)$ for $h$ above, $(X_1, \hat{\theta}_n) \xrightarrow{p} (X_1, \theta^\ast)$ implies $\hat{\gamma}_n(X_1) \xrightarrow{p} \gamma^\ast(X_1)$ and $\|\hat{\gamma}_n(X_1) - \gamma^\ast(X_1)\|_1 \xrightarrow{p} 0$ by the continuous mapping theorem. Hence, we have $\mathsf{E}_n(\gamma) = \E \|\hat{\gamma}_n(X_1) - \gamma^\ast(X_1)\|_1 \to 0$ as in the proof of Theorem \ref{thm:asymptotic_coverage}.

    For the stability term, we again prove \eqref{eq:hellinger_sum_convergence} by showing \eqref{eq:marginal_stability}. Now, let $\hat{\theta}_n'$ be a consistent root of the likelihood equations based on $X_{n + 1}, X_2, \ldots, X_n$. Then, $\hat{\gamma}_{1 \to n + 1}(X_j) = h(X_j, \hat{\theta}_n')$. From \eqref{eq:asymptotic_linearity}, we have 
    \begin{equation*}
        \hat{\theta}_n' - \hat{\theta}_n = \frac{1}{n} (\psi(X_{n + 1}) - \psi(X_1)) + O_p(n^{-1}) = O_p(n^{-1}).
    \end{equation*} 
    Next, by Lemma \ref{lem:conditional_density_distance_approx}, we have
    \begin{equation*}
        \begin{split}
            H^2\left(\hat{\gamma}_n(X_2), \hat{\gamma}_{1 \to n + 1}(X_2)\right) 
            & = H^2(h(X_2, \hat{\theta}_n), h(X_2, \hat{\theta}_n')) \\
            & = \frac{1}{8} \langle \hat{\theta}_n' - \hat{\theta}_n, I_{X_2}(\hat{\theta}_n) (\hat{\theta}_n' - \hat{\theta}_n) \rangle + o(\|\hat{\theta}_n' - \hat{\theta}_n\|^2).
        \end{split}
    \end{equation*}
    Since $\hat{\theta}_n' - \hat{\theta}_n = O_p(n^{-1})$, we have $o(\|\hat{\theta}_n' - \hat{\theta}_n\|^2) = o_p(n^{-2})$ by Lemma 2.12 of \citet{vaart_1998}. Now, consider the vectorization of $I_x(\theta)$ and the Jacobian $J_x(\theta)$ of $\theta \mapsto I_x(\theta)$. Since $\theta \mapsto I_x(\theta)$ is smooth, the following first-order approximation holds:
    \begin{equation*}
        I_{X_2}(\hat{\theta}_n) = I_{X_2}(\theta^\ast) + J_{X_2}(\theta^\ast) (\hat{\theta}_n - \theta^\ast) + o(\|\hat{\theta}_n - \theta^\ast\|).
    \end{equation*}
    As $\hat{\theta}_n - \theta^\ast = O_p(n^{-1 / 2})$, we have $o(\|\hat{\theta}_n - \theta^\ast\|) = o_p(n^{-1 / 2})$ by Lemma 2.12 of \citet{vaart_1998}. Hence,
    \begin{equation*}
        I_{X_2}(\hat{\theta}_n) = I_{X_2}(\theta^\ast) + J_{X_2}(\theta^\ast) (\hat{\theta}_n - \theta^\ast) + o_p(n^{-1 / 2}) = I_{X_2}(\theta^\ast) + o_p(1),
    \end{equation*}
    where the last equality follows as $J_{X_2}(\theta^\ast) = O_p(1)$ and $\hat{\theta}_n - \theta^\ast = o_p(1)$. Hence, $I_{X_2}(\hat{\theta}_n) = O_p(1)$. Accordingly, we have
    \begin{equation*}
        \left|\left\langle \hat{\theta}_n' - \hat{\theta}_n, I_{X_2}(\hat{\theta}_n) (\hat{\theta}_n' - \hat{\theta}_n) \right\rangle\right| 
        \le \|I_{X_2}(\hat{\theta}_n)\|_{\mathrm{op}} \cdot \|\hat{\theta}_n' - \hat{\theta}_n\|^2
        = O_p(1) \times O_p(n^{-2})
        = O_p(n^{-2}),
    \end{equation*}
    where $\|\cdot\|_{\mathrm{op}}$ is the operator norm. Since $O_p(n^{-2})$ is also $o_p(n^{-1})$, we have \eqref{eq:marginal_stability}. Hence, as in the proof of Theorem \ref{thm:asymptotic_coverage}, we have \eqref{eq:hellinger_sum_convergence} and thus $\mathsf{S}_n(\gamma) \to 0$.
\end{proof}

\subsection{Auxiliary Lemmas}
\begin{lemma}
    \label{lem:tv_conditional}
    Let $\cX, \cY$ be two measurable spaces. Let $(X, Y)$ and $(X', Y')$ be random variables taking values in the product measurable space $\cX \times \cY$, whose distributions are denoted as $P$ and $P'$, respectively. Suppose that the space $\cY$ is regular\footnote{For instance, see Theorem 2.18 of Chapter 4 of \citet{cinlar_2011}.} so that the conditional laws $Y \, | \, X = x$ and $Y' \, | \, X' = x$ exist, which we denote as $P |_x$ and $P' |_x$, respectively. If $X$ and $X'$ have the same marginal distribution, say, $P_X$ on $\cX$, then we have
    \begin{equation}
        \label{eq:tv_conditional_bound}
        \mathsf{TV}(P, P') \le \int_{\cX} \mathsf{TV}(P |_x, P' |_x) \, \mathrm{d} P_X(x).
    \end{equation}
\end{lemma}
\begin{proof}
    Let $\cM$ be the collection of all measurable subsets of $\cX \times \cY$. For any $A \in \cM$ and $x \in \cX$, let $A_x = \{y \in \cY : (x, y) \in A\}$. Then, 
    \begin{equation*}
        \begin{split}
            P(A) 
            & = \int_{\cX \times \cY} 1\{(x, y) \in A\} \, \mathrm{d} P(x, y) \\
            & = \int_{\cX} \int_{\cY} 1\{y \in A_x\} \, \mathrm{d} P |_x(y) \mathrm{d} P_X(x) \\
            & = \int_{\cX} P |_x(A_x) \, \mathrm{d} P_X(x).
        \end{split}
    \end{equation*}
    Similarly, $P'(A) = \int_{\cX} P'|_x(A_x) \, \mathrm{d} P_X(x)$. Hence, for any $A \in \cM$, 
    \begin{equation*}
        \begin{split}
            |P(A) - P'(A)| 
            & = \left|\int_{\cX} \left(P |_x(A_x) - P' |_x(A_x)\right) \, \mathrm{d} P_X(x)\right| \\
            & \le \int_{\cX} \left|P |_x(A_x) - P' |_x(A_x)\right| \, \mathrm{d} P_X(x) \\
            & \le \int_{\cX} \mathsf{TV}(P |_x, P' |_x) \, \mathrm{d} P_X(x).
        \end{split}
    \end{equation*}
    As this is true for any $A \in \cM$, we have \eqref{eq:tv_conditional_bound}.
\end{proof}

\begin{lemma}
    \label{lem:sum_convergence_to_log}
    Consider a triangular array of random variables $\{(a_{n 1}, \ldots, a_{n n}) : n \in \N\}$, where $a_{n i} \in [0, 1]$ for all $n, i$. Let $S_n = \sum_{i = 1}^{n} a_{n i}$ and $L_n := \sum_{i = 1}^{n} \log(1 - a_{n i})$. Then, $S_n \xrightarrow{p} 0$ implies $L_n \xrightarrow{p} 0$.
\end{lemma}
\begin{proof}
    Fix $\epsilon > 0$. We need to show $\lim_{n \to \infty} \P(|L_n| > \epsilon) = 0$. One can verify that $\log(1 - z) \ge -z - z^2$ for $z \in [0, \frac{1}{2}]$. If $S_n \le \frac{1}{2}$, we have $a_{n 1}, \ldots, a_{n n} \le \frac{1}{2}$ and thus $L_n \ge -S_n - \sum_{i = 1}^{n} a_{n i}^2 \ge - 2 S_n$, where the last inequality follows from $a_{n i} \in [0, 1]$. Therefore, $(|L_n| > \epsilon) \cap \left(S_n \le \frac{1}{2}\right)$ implies $-2 S_n < -\epsilon$. Hence, 
    \begin{equation*}
        \begin{split}
            \P(|L_n| > \epsilon) 
            & = \P\left((|L_n| > \epsilon) \cap \left(S_n > \frac{1}{2}\right)\right) + \P\left((|L_n| > \epsilon) \cap \left(S_n \le \frac{1}{2}\right)\right) \\
            & \le \P\left(S_n > \frac{1}{2}\right) + \P\left(S_n > \frac{\epsilon}{2}\right).
        \end{split}
    \end{equation*}
    Since $S_n \xrightarrow{p} 0$, we conclude $\lim_{n \to \infty} \P(|L_n| > \epsilon) = 0$.
\end{proof}
\section{Additional Empirical Results}
\label{sec:additional_empirical_results}

\subsection{Simulations with Stochastic Fuzzy-C-Means Clustering}
We repeat the two-dimensional Gaussian Mixture Model (GMM) simulation in Section \ref{sec:simulations} with stochastic fuzzy-c-means (FCM) clustering instead of stochastic GMM clustering. FCM is a popular clustering algorithm that allows for soft cluster assignments, and its stochastic version can be implemented by sampling the cluster labels according to the soft assignments as in Definition \ref{def:stochastic_clustering}. The fuzziness parameter $m$ in FCM controls the degree of softness in the cluster assignments, with larger $m$ leading to softer assignments, while $m \to 1$ corresponds to hard clustering similar to K-means. We consider three values of $m \in \{1.4, 1.7, 2.0\}$ to represent different levels of softness in the cluster assignments. 

\begin{figure}[!ht]
    \centering
    \includegraphics[width=0.75\textwidth]{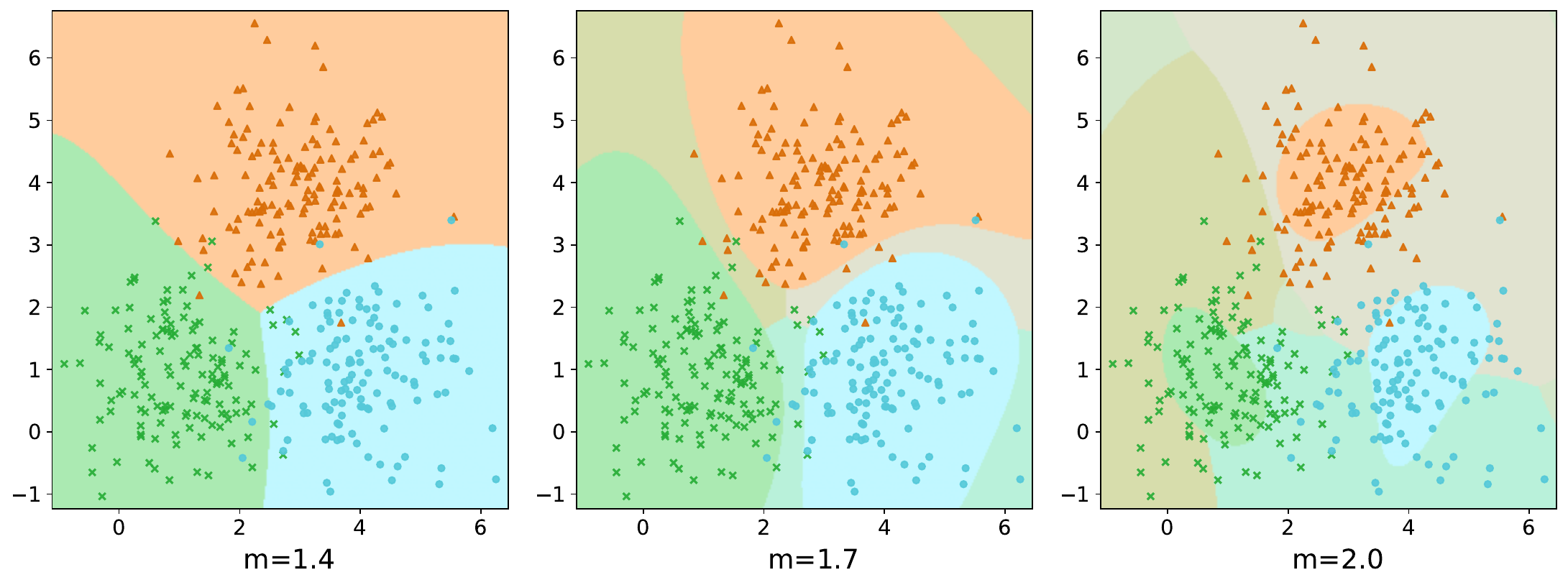}

    \vspace*{10pt}

    \includegraphics[width=\textwidth]{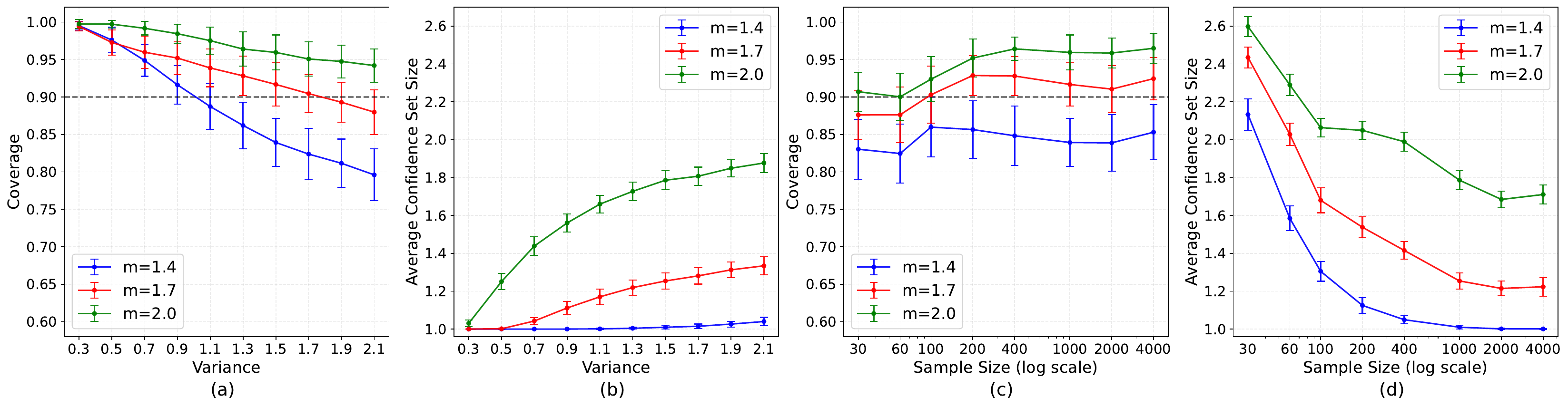}
    \caption{Confidence set heatmaps, coverage, and confidence set size for GMM simulations with stochastic FCM clustering with the fuzziness parameter $m \in \{1.4, 1.7, 2.0\}$. The heatmaps are for the same setting as in the left-most panel (small variance) of Figure \ref{fig:heatmaps}. (a) and (b) show the results for varying $\sigma^2$ for fixed sample size $n = 1000$, while (c) and (d) are for fixed variance $\sigma^2 = 1.5$ and increasing sample size $n$. }
    \label{fig:GMM_FCM}
\end{figure}

First, let us compare the confidence set heatmaps in the top row of Figure \ref{fig:GMM_FCM} to the left most panel of Figure \ref{fig:heatmaps} based on stochastic GMM clustering. Clearly, $m = 1.4$ leads to no uncertainty in the cluster labels, while $m = 2.0$ ends up with excessively uncertain cluster labels. Though $m = 1.7$ shows more uncertainty than stochastic GMM clustering, the overall pattern suggests that this choice of $m$ leads to a practically useful level of uncertainty in the cluster labels.

The bottom row of Figure \ref{fig:GMM_FCM} shows the coverage and confidence set size for varying $\sigma^2$ and $n$, equivalent to the setting of the top row of Figure \ref{fig:GMM}. As hinted by the heatmaps, we confirm that $m = 1.4$ leads to under-coverage like the naive split conformal clustering (Algorithm \ref{alg:split_conformal_clustering} with hard GMM clustering), while $m = 2.0$ leads to over-coverage like the naive cutoff method (Algorithm \ref{alg:cutoff}). On the other hand, $m = 1.7$ shows a good balance between coverage and confidence set size, even though the coverage drops further as $\sigma^2$ increases. Overall, these results show that a reasonable choice of the fuzziness parameter $m$ in stochastic FCM clustering can perform well in practice, which opens up a promising avenue for future research on the choice of $m$ and its theoretical properties in the context of conformal inference for clustering.

\subsection{APOGEE Data}
We next apply our method to the APOGEE dataset, which has been studied extensively in the astronomical literature, often with conflicting conclusions regarding the true number of clusters and the feasibility of confident clustering \citep{casamiquela2021impossibility,chen2018chemodynamical,ratcliffe2020tracing,pagnini2025abundance,berni2024searching}. We follow the preprocessing and validation philosophy of \citet{chang2025unsupervised}, which uses the APOGEE DR17 value-added catalog of globular-cluster stars as its starting point and emphasizes systematic comparison across preprocessing, dimension-reduction, and clustering choices. In particular, from the abundance sets considered there, we choose \(11\) features, while using their reported labels as a reference partition due to the lack of ground-truth labels. Their analysis suggests that K-means with $K = 8$ clusters provides a stable/generalizable target, and we adopt their recommended t-SNE projection (perplexity 100) for low-dimensional visualization.

Because K-means is a hard-label algorithm, it is not suitable to reflect uncertainty. Instead, we employ a stochastic version of fuzzy-c-means, a soft clustering method that reduces to K-means as its fuzziness parameter approaches 1. We choose a fuzziness parameter relatively close to 1 to preserve the qualitative behavior of K-means while enabling stochastic label generation, and use a random forest classifier on the training set.

The confidence sets produced in the original feature space are visualized in the 2D t-SNE space (Figure \ref{fig:astro_plot_fcms}, top). The left plot displays the reference labels. The middle plot shows the cluster assignments predicted by the stochastic fuzzy-c-means procedure for the training and calibration sets. The right plot presents the resulting confidence sets for the test points.

\begin{figure}[!ht]
    \centering
    \subfloat{\includegraphics[width=\linewidth]{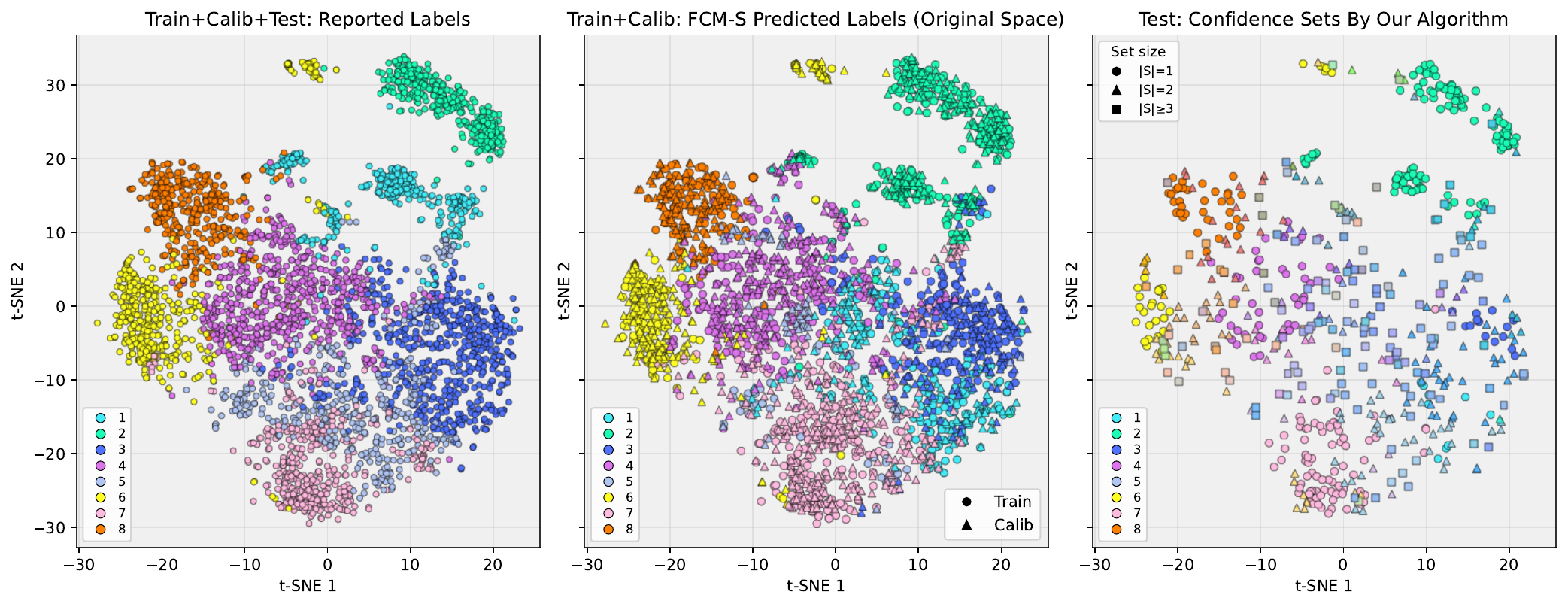}} \\
    \subfloat{\includegraphics[width=0.67\linewidth]{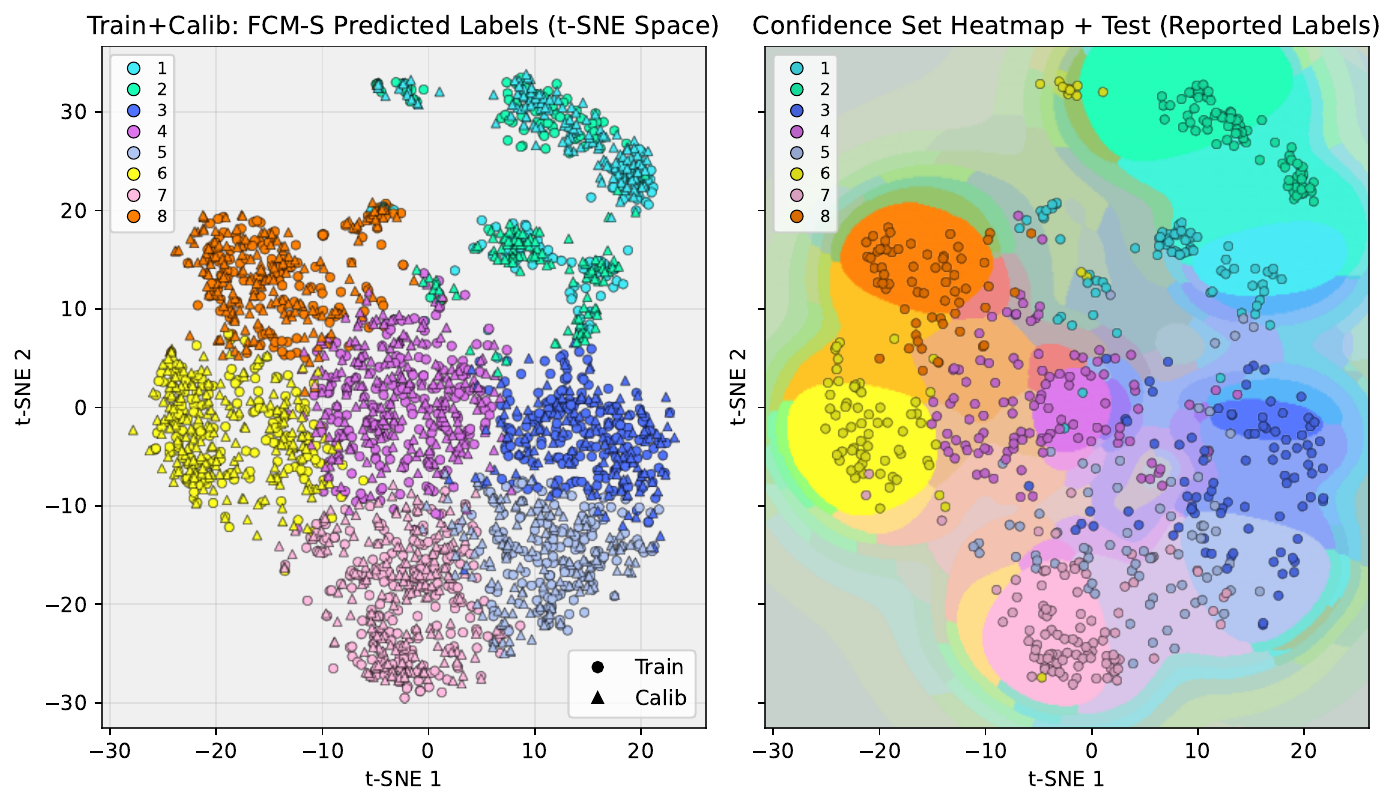}}
    \caption{Application of Algorithm \ref{alg:split_conformal_clustering_stochastic} to the APOGEE dataset with stochastic fuzzy-c-means. Top row: our method is applied in the original feature space and the outputs are visualized in the two-dimensional t-SNE space. From left to right, the plots show the reported reference labels, the predicted cluster labels, and the confidence sets for the test points. Bottom row: our method is applied directly in the two-dimensional t-SNE space. The left plot shows the predicted cluster labels, and the right plot shows a heatmap of the confidence sets over the projected space with the test points overlaid using their reported labels. As in Figure \ref{fig:pbmc_plots}, marker shape encodes the confidence set size, and colors indicate the cluster labels contained in the confidence set.}
    \label{fig:astro_plot_fcms}
\end{figure}

As before, marker shapes indicate the confidence set size (circles for singletons, triangles for pairs, squares for size $\ge 3$), and colors are blended to reflect the contained labels, with lighter shades indicating greater uncertainty. We observe that Clusters 2, 6, 7, and 8 tend to retain more singleton outputs, which directly aligns with the findings in \citet{chang2025unsupervised}---see their Figure 4C---identifying these specific clusters as having the highest local generalizability.

For completeness, we also apply and visualize our method directly in the 2D t-SNE space using the same prior choice of $K$, the stochastic fuzzy-c-means model, and a random forest classifier (Figure \ref{fig:astro_plot_fcms}, bottom). The left plot displays the predicted cluster assignments for the training and calibration observations. The right plot presents a heatmap of the confidence sets over the projected domain, constructed and colored as before, with overlaid test observations. As reiterated previously, while the heatmap illustrates the behavior of prediction sets in the projected space, such visualizations should be interpreted cautiously; scientifically robust conclusions must rely on the high-dimensional original space analysis.

Although \citet{chang2025unsupervised} recommends K-means as the best clustering algorithm for this dataset, we also investigate the results when a different clustering algorithm is applied. This serves primarily to demonstrate the loss of interpretability in downstream tasks when an inappropriate clustering algorithm is used. To this end, we apply our method using stochastic GMM clustering. The results are presented in Figure \ref{fig:astro_plot_gmms}. 

As before, the first panel of the top row of Figure \ref{fig:astro_plot_gmms} shows the output of our algorithm applied in the original feature space ($p = 11$), visualized via a two-dimensional t-SNE projection. The second panel displays the cluster labels predicted by the stochastic GMM clustering algorithm on the training and calibration sets, corresponding to Steps 2 and 4 of Algorithm \ref{alg:split_conformal_clustering_stochastic}. The right panel presents the confidence sets produced by Algorithm \ref{alg:split_conformal_clustering_stochastic} for the test points in the original space. The markers are coded by the size of the confidence set, and their colors correspond to the cluster labels contained in the confidence sets. The bottom row shows the results of our algorithm when applied directly to the data in the projected space using stochastic GMM clustering. The left panel shows the predicted labels for the training and calibration sets, while the right panel provides a heatmap of the confidence sets for points across the projected space.

\begin{figure}[!ht]
    \centering
    \subfloat{\includegraphics[width=\linewidth]{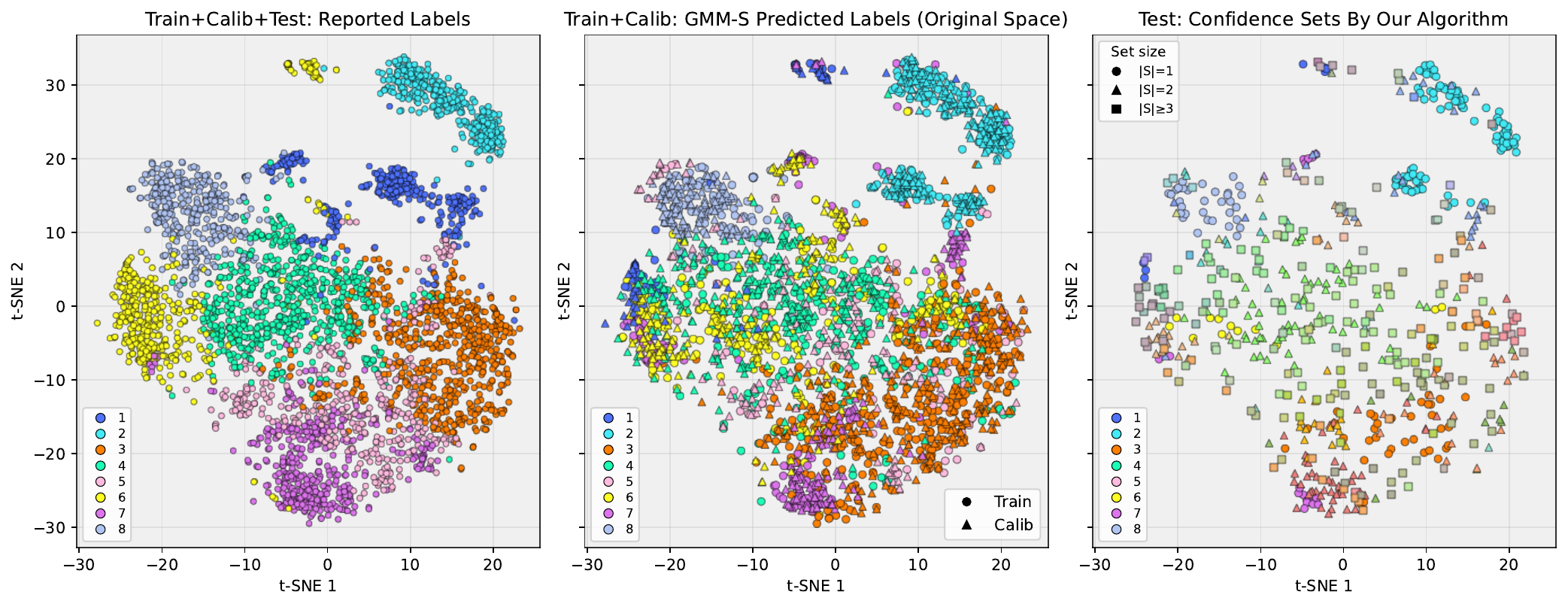}} \\
    \subfloat{\includegraphics[width=0.67\linewidth]{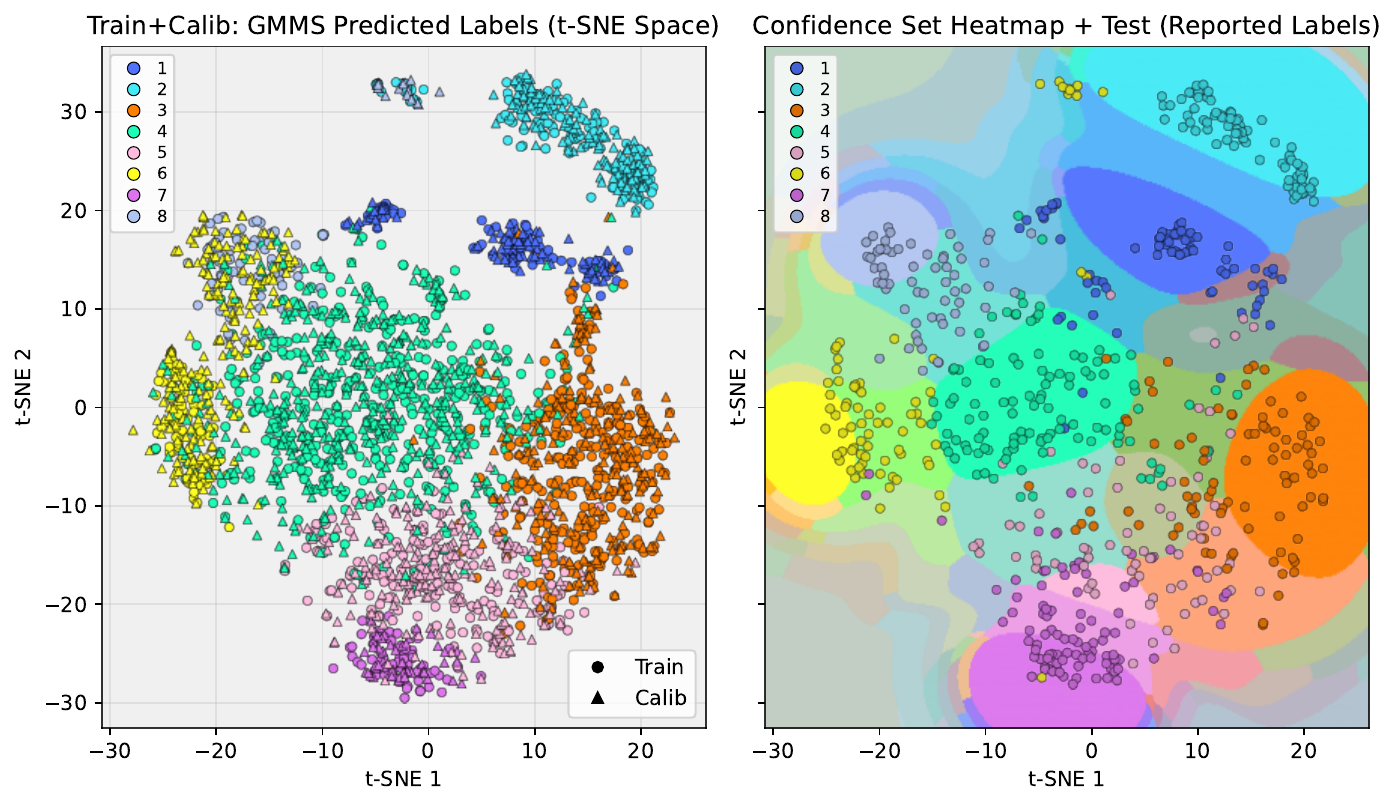}}
    \caption{Application of Algorithm \ref{alg:split_conformal_clustering_stochastic} to the APOGEE dataset with stochastic GMM clustering. Top row: our method is applied in the original feature space and the outputs are visualized in the two-dimensional t-SNE space. From left to right, the plots show the reported reference labels, the predicted cluster labels, and the confidence sets for the test points. Bottom row: our method is applied directly in the two-dimensional t-SNE space. The left plot shows the predicted cluster labels, and the right plot shows a heatmap of the confidence sets over the projected space with the test points overlaid using their reported labels. As in Figure \ref{fig:pbmc_plots}, marker shape encodes the confidence set size, and colors indicate the cluster labels contained in the confidence set.}
    \label{fig:astro_plot_gmms}
\end{figure}

Comparing these results to Figure \ref{fig:astro_plot_fcms}, it is evident that the Gaussian mixture model, which is ill-suited for this dataset, fails to confidently recover the otherwise stable and generalizable clusters. In fact, the only cluster label exhibiting high certainty (singleton prediction sets denoted by circles) is Cluster 2, which does not constitute a particularly meaningful finding for this dataset. Consequently, the resulting confidence sets are difficult to interpret for downstream analysis. Furthermore, in the projected space, although the algorithm yields visually coherent regions, the intrinsic properties of the original data are no longer preserved. Thus, while applying the method in a projected space might seem appealing for visualization purposes, the resulting interpretations might not be reliably translated back to the original dataset to drive scientific discovery.

\renewcommand{\refname}{Appendix References}
\putbib
\end{bibunit}

\end{document}